\documentclass[10pt,twocolumn,letterpaper]{article}
\usepackage{iccv}
\definecolor{myblue}{rgb}{0.21,0.49,0.74}
\usepackage[pagebackref,breaklinks,colorlinks,allcolors=myblue]{hyperref}

\newcommand{\myparatight}[1]{\smallskip\noindent{\bf {#1}:}~}

\usepackage{times}
\usepackage{amsthm}
\usepackage{amsmath}
\usepackage{amssymb}
\usepackage{float}
\usepackage{graphics}                        
\usepackage{graphicx}
\usepackage[skip=0pt]{caption}
\usepackage{algorithm}
\usepackage{algpseudocode}
\usepackage{bm}
\usepackage{color}
\usepackage{multirow}
\usepackage{makecell}
\usepackage{soul}
\usepackage{pifont}
\usepackage{xcolor}

\usepackage{color, colortbl}
\definecolor{greyL}{RGB}{230,248,255}

\DeclareMathOperator*{\argmax}{argmax}

\usepackage{wrapfig}

\allowdisplaybreaks
\usepackage{xspace}

\usepackage[scaled]{beramono}

\newcommand{\alg}{\texttt{ATM}\xspace}

\newtheorem{thm}{Theorem}
\newtheorem{lem}{Lemma}

\title{Find a Scapegoat: Poisoning Membership Inference Attack and Defense to Federated Learning}

\author{Wenjin Mo$^{1}\thanks{Equal contribution. Wenjin Mo and Zhiyuan Li conducted this research while they were interns under the supervision of Minghong Fang.}$ \quad Zhiyuan Li$^{2*}$  \quad Minghong Fang$^{3}\thanks{Corresponding author.}$ \quad Mingwei Fang$^{4\dagger}$\\
$^{1}$Yale University, $^{2}$Independent Researcher, \\
$^{3}$University of Louisville, $^{4}$Guangdong Polytechnic Normal University \\
\texttt{wenjin.mo@yale.edu, arin.lee.lzy@gmail.com}, \\
\texttt{minghong.fang@louisville.edu, mw.fang@gpnu.edu.cn}
}

\begin{document}

\maketitle

\begin{abstract}
Federated learning (FL) allows multiple clients to collaboratively train a global machine learning model with coordination from a central server, without needing to share their raw data. This approach is particularly appealing in the era of privacy regulations like the GDPR, leading many prominent companies to adopt it. However, FL's distributed nature makes it susceptible to poisoning attacks, where malicious clients, controlled by an attacker, send harmful data to compromise the model. Most existing poisoning attacks in FL aim to degrade the model's integrity, such as reducing its accuracy, with limited attention to privacy concerns from these attacks. In this study, we introduce FedPoisonMIA, a novel poisoning membership inference attack targeting FL. FedPoisonMIA involves malicious clients crafting local model updates to infer membership information. Additionally, we propose a robust defense mechanism to mitigate the impact of FedPoisonMIA attacks. Extensive experiments across various datasets demonstrate the attack's effectiveness, while our defense approach reduces its impact to a degree.
\end{abstract}


\section{Introduction} \label{sec:intro}

Federated Learning (FL)~\cite{konevcny2016federated1, konevcny2016federated2, yang2019federated, McMahan17} is a decentralized machine learning framework that allows multiple clients to collaboratively train a shared global model while maintaining data privacy by keeping raw data localized. In FL, the central server initiates the process by distributing initial global model parameters to all participating clients. Each client subsequently performs local model training on its private dataset, generating model updates that are then transmitted back to the server, where they are aggregated according to predefined aggregation rules. The server then updates the global model with the aggregated update and redistributes it to the clients. This iterative cycle continues until the model converges. 
Due to its emphasis on privacy preservation, FL has gained widespread adoption. However, recent research~\cite{yin2024poisoning, biggio2012poisoning, kairouz2021advances, bagdasaryan2020backdoor, baruch2019little, fang2020local, li2022learning, shejwalkar2021manipulating, tolpegin2020data,wang2025poisoning,zhang2024poisoning,yin2024poisoning} highlights vulnerabilities in FL to poisoning attacks, where malicious clients may send carefully crafted updates to alter the performance of the global model. Among these attacks, a specific type known as the \textit{poisoning membership inference attack} (PMIA)~\cite{choquette2021label, shokri2017membership, yeom2018privacy, wen2024privacy, ma2023loden, mahloujifar2022property, carlini2022membership, truex2018towards, ye2022enhanced} enables malicious clients to deduce whether a particular data sample is included in the training data of other clients, thereby threatening the privacy integrity of the FL system.

In this study, we introduce a sophisticated and novel attack method, which we call FedPoisonMIA, designed to surpass the capabilities of existing methods and expose critical, previously unaddressed privacy risks within FL systems. This method carefully crafts malicious model updates that maximize angular deviation relative to standard benign updates, which in turn escalates the risk of privacy breaches within the FL environment. 
For instance, FL is widely used in healthcare, allowing hospitals to train a shared model. However, our attack can be leveraged to extract sensitive patient information in this setting.
The primary goal of FedPoisonMIA is to exploit this deviation to infiltrate the FL process while minimizing the risk of detection. To achieve this, the attack method carefully embeds its malicious updates within a collection of benign updates, effectively disguising them to evade detection and filtering mechanisms that are conventionally employed by the central server. This strategic concealment not only ensures that the attack remains undetected over multiple communication rounds but also preserves its capacity to undermine privacy across the entire FL process, enabling persistent and ongoing privacy compromise. Through this technique, FedPoisonMIA reveals the limitations of current FL defenses and highlights the pressing need for advanced protection mechanisms against such nuanced and deeply embedded attacks.

While a variety of Byzantine-robust mechanisms have been proposed to counteract the adverse effects of poisoning attacks in FL~\cite{McMahan17, yin2018byzantine, aji2017sparse, dwork2008differential, blanchard2017machine, fang2020local,2021FLTrust,fang2025byzantine,fang2025we,fang2024byzantine,fang2025provably,fang2022aflguard,xie2024fedredefense,xu2024robust,xie2018generalized}, these methods predominantly focus on preventing data and model corruption. However, they largely overlook membership inference attacks, a distinct and persistent threat that seeks to uncover information about the participation of individual data points in the training process. Existing Byzantine-robust approaches have shown limited efficacy in addressing this privacy vulnerability, particularly against our newly introduced attack method, which strategically maximizes angular deviation to evade detection.
To address this critical gap in FL defenses, we propose a novel Byzantine-robust mechanism named Angular Trimmed-mean (\alg), designed specifically to counteract such membership inference attacks with heightened resilience. Our method employs angular deviation criteria to rigorously scrutinize incoming model updates, identifying and filtering out malicious contributions based on their deviation from the majority's directional alignment. Specifically, updates exhibiting substantial angular deviations from the bulk of other updates are flagged as potential threats and subsequently removed from the aggregation process. 
This approach could effectively mitigate the impact of malicious clients.
By implementing \alg, we aim to reinforce FL’s robustness against privacy breaches, bridging the existing gap in protection against membership inference vulnerabilities.

Experimental evaluations conducted on a diverse set of datasets from various domains reveal that our proposed attack method is capable of consistently bypassing the detection measures of all examined Byzantine-robust mechanisms, while also achieving a notably high attack accuracy. These results underscore the effectiveness of our attack in navigating around existing defenses, thereby highlighting a significant privacy vulnerability within federated learning systems. Conversely, our experimental findings further demonstrate the efficacy of our defense mechanism in counteracting multiple types of PMIAs. By successfully reducing the attack accuracy of these PMIAs, our defense approach plays a crucial role in diminishing the associated risks of privacy leakage. This twofold experimental analysis emphasizes the need for improved defense mechanisms and showcases the capability of our proposed \alg to mitigate privacy risks within FL frameworks.

Our main contributions are as follows:
\begin{itemize}
    \item
    We introduce an innovative PMIA method that enhances the angular deviation between malicious and benign updates to maximize impact while evading detection, resulting in high attack accuracy against a range of established Byzantine-robust mechanisms.
    
    \item 

    We present a Byzantine-robust defense mechanism called \alg, designed to detect and filter malicious updates by assessing the angular distance between them. This approach effectively reduces the impact of PMIA attacks within FL systems.

    \item 
Our experiments on diverse benchmarks confirm that our attack outperforms existing PMIAs against various Byzantine-robust defenses. Additionally, our proposed defense effectively reduces PMIA accuracy, substantially enhancing privacy in FL.

\end{itemize}

\section{Preliminaries and Related Work}
\label{sec:preliminaries}
\subsection{Federated Learning (FL): An Overview} 
\label{sec:background}

Consider a federated learning (FL) system with \( n \) clients and a central server. Each client \( k \in [n] \) possesses a local dataset \( \mathcal{D}_k \). Let \( \mathcal{D} \) represent the combined dataset of all clients, defined as \( \mathcal{D} = \cup_{k \in [n]} \mathcal{D}_k \). 
Rather than training a machine learning model on the entire dataset \( \mathcal{D} \), FL allows these \( n \) clients to collaboratively train a single global model with the support of the central server, without sharing each client’s raw training data.
The training objective in FL can be formulated as the following optimization problem:
\begin{align}
\label{fl_obj}
\min_{\bm{w} \in \mathbb{R}^d} f(\bm{w}) =  \sum_{k \in [n]} \frac{|\mathcal{D}_k|}{|\mathcal{D}|} F_k(\bm{w}, \mathcal{D}_k),
\end{align}
where \( \bm{w} \) represents the model parameters, \( d \) is the dimension of \( \bm{w} \), \( |\mathcal{D}_k| \) denotes the size of \( \mathcal{D}_k \), and \( F_k(\bm{w}, \mathcal{D}_k) \) is the local training objective for client \( k \). 
In particular, FL tackles Problem~(\ref{fl_obj}) through an iterative process. During training round \( t \), this involves the following three steps:
\begin{itemize}

    \item \myparatight{Global Model Synchronization}The central server selects a fraction \( C \) of the clients and sends the current global model \( \bm{w}^t \) to each of these chosen clients, where \( C \) falls within the range \( (0, 1] \).

    \item \myparatight{Training of local models}Each selected client \( k \) refines its local model using the current global model \( \bm{w}^t \) along with its local dataset. Specifically, client \( k \) selects a mini-batch of training example \( \mathcal{S}_k \) from \( \mathcal{D}_k \) and calculates a local model update in the form of a gradient \( \bm{g}_k^t = \frac{1}{|\mathcal{S}_k|} \sum_{h \in \mathcal{S}_k} \nabla F_k(\bm{w}^t, h) \). 
    The computed update \( \bm{g}_k^t \) is then transmitted to the server.

    \item \myparatight{Updating of the global model}After receiving the local model updates from all clients, the server applies an aggregation rule \( \mathcal{A} \) to merge these updates. It then updates the global model as follows:
    \begin{align}
    \label{fl_model_aggre}
    \bm{w}^{t+1} = \bm{w}^t - \eta \cdot \mathcal{A}(\{\bm{g}_k^t\}_{k \in [n]}),
    \end{align}
    where $\eta$ is the learning rate and assume $C=1$ in Eq.~(\ref{fl_model_aggre}). FL methods mainly differ in their aggregation rules. For example, FedAvg~~\cite{McMahan17} aggregates updates as \( \mathcal{A}(\{\bm{g}_k^t\}_{k \in [n]}) = \sum_{k \in [n]} \frac{|\mathcal{D}_k|}{|\mathcal{D}|} \bm{g}_k^t \).
\end{itemize}

\subsection{Membership Inference Attacks to FL}
\label{sec:Privacy_Attacks_FL}

Membership Inference Attacks (MIA)~\cite{hu2022membership} are privacy attacks that seek to determine if a given input sample is part of a target machine learning model’s training data. MIA can be classified as either passive or active. In passive attacks, an attacker queries a trained model through an API and identifies training samples by analyzing the model’s responses; for example, prior work~\cite{leino2020stolen} assumes that a sample is a ``member'' if the model’s prediction is accurate, suggesting higher accuracy on familiar data. Similarly, high prediction confidence has been used as an indicator that a sample is part of the training set~\cite{song2019privacy,yeom2018privacy}. In contrast, active attacks~\cite{gomes2024active,chen2020gan} involve manipulating the attacker's local data or directly crafting gradient updates. For instance,~\cite{nasr2019comprehensive} uses gradient ascent to heighten the response difference between trained and untrained data, while AGREvader~\cite{zhang2023agrevader} masks gradients from label-flipped samples to better evade Byzantine-robust defenses.

\myparatight{Distinctions between poisoning attacks vs poisoning membership inference attacks (PMIAs)}%
Poisoning attacks in FL seek to degrade global model integrity by introducing harmful data or gradient updates, resulting in reduced classification accuracy. Conversely, PMIAs target client privacy by identifying the presence of particular data points within benign clients' datasets, without noticeably impacting the global model's performance. This subtlety renders PMIAs design inherently more difficult compared to conventional poisoning attacks.

\subsection{Defenses Against MIA to FL} 
\label{sec:Robust_FL}

Defenses against MIA in FL can be categorized into non-aggregation-based and aggregation-based approaches. Among non-aggregation methods, Differential privacy (DP)~\cite{dwork2008differential} and Top-$k$~\cite{aji2017sparse} are prominent. DP adds Gaussian noise to gradients to reduce privacy risks, while Top-$k$ selects only the top $k$ dimensions with the highest absolute values in each gradient, nullifying others to minimize attack effects.
In aggregation-based defenses, several Byzantine-robust rules have been proposed~\cite{rieger2022deepsight,yin2018byzantine,blanchard2017machine,fang2020local}. 
For instance,
the Median~\cite{yin2018byzantine} method calculates the element-wise median of client updates, resisting outliers, though it may fail when malicious updates resemble benign ones.


\begin{table}[t]
  \centering
  \scriptsize
   \addtolength{\tabcolsep}{-2.75pt}
  \caption{Difference between full-knowledge attack and partial-knowledge attack.}
    \begin{tabular}{|c|c|c|c|}
    \hline
    & {\makecell {Benign clients' \\ gradients}}  & {\makecell {Malicious clients' \\ gradients}}  & {\makecell {Server' \\ aggregation rule}}  \\
    \hline
    Full-knowledge attack & \textcolor{green!70!black}{\ding{51}}  &  \textcolor{green!70!black}{\ding{51}} & \textcolor{red!90!black}{\ding{55}} \\
    \hline
    Partial-knowledge attack & \textcolor{red!90!black}{\ding{55}}  & \textcolor{green!70!black}{\ding{51}} & \textcolor{red!90!black}{\ding{55}} \\
    \hline
    \end{tabular}%
  \label{tab:full_part}%
  		\vspace{-.2in}
\end{table}%

\section{Problem Statement} \label{sec::problem Statement}
\myparatight{Attacker's goal}The objective of the attacker is to infer indirectly whether particular samples are part of the training sets used by benign clients within the FL system, effectively enabling a form of data theft. This type of attack is particularly concerning as it reveals sensitive information about the clients’ private datasets without requiring direct access. By confirming the presence of specific samples, the attacker can breach the confidentiality of client data, undermining the core privacy protections FL is designed to provide.

\myparatight{Attacker’s capability and knowledge}In line with previous works~\cite{sun2019can, munoz2017towards, tolpegin2020data, bhagoji2019analyzing}, our attack model allows the attacker to manipulate its local data and adjust its updates before sending them to the server. Additionally, as per~\cite{fang2020local,shejwalkar2021manipulating}, the attacker may control multiple malicious clients/devices. 
In line with previous studies~\cite{li2024open, fang2020local, yin2018byzantine}, we analyze two levels of attacker knowledge: full-knowledge attack, where the attacker leverages updates from all clients (including benign ones) to design malicious updates, and partial-knowledge attack, where the attacker uses only malicious clients' updates. 
In both scenarios, the attacker is unaware of the server's aggregation method. Table~\ref{tab:full_part} provides an overview of these attack scenarios.

\myparatight{Defender's goal}Our defense achieves Byzantine robustness against malicious clients, protecting benign clients' privacy while preserving model accuracy and efficiency. Specifically, we aim for: (1) Robustness, minimizing attackers' ability to steal local data; (2) Fidelity, ensuring accuracy comparable to FedAvg when no attacks occur; and (3) Efficiency, maintaining client workloads similar to FedAvg without extra computational overhead.

\section{Our Attack} 
\label{sec:alg}

\begin{figure}[t]
	\centering
	\includegraphics[scale = 0.42]{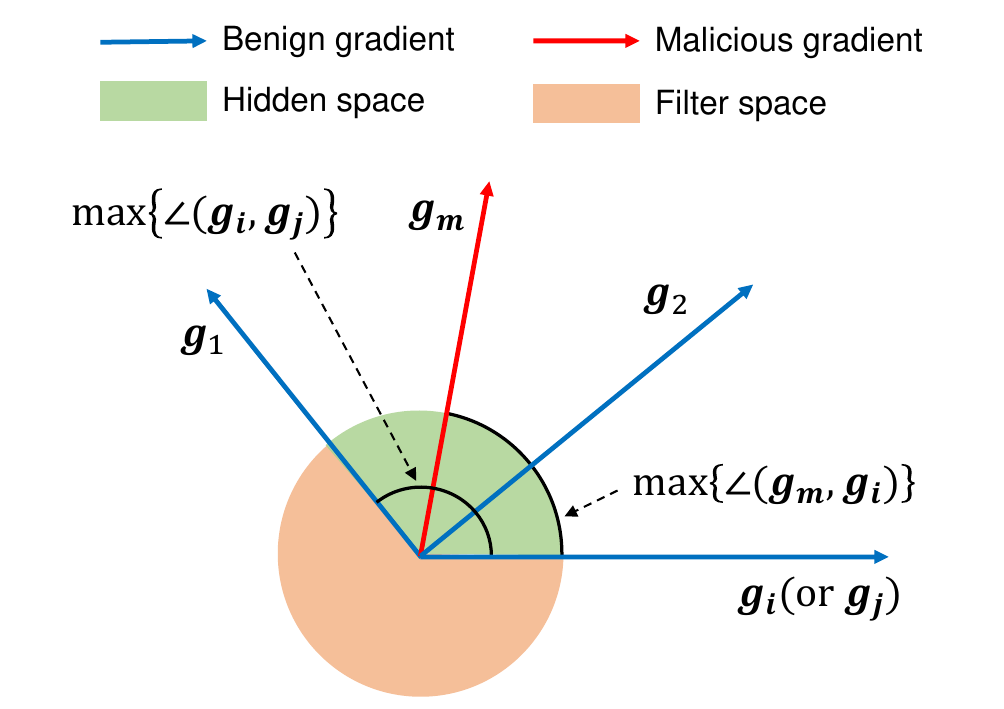}
	\caption{Overview of our attack: Malicious gradients target specific samples while blending into benign updates. By ensuring their angular deviation stays below the largest benign gradient difference, our attack manipulates robust defenses into mistakenly discarding benign gradients (e.g., $\bm{g}_1$) as outliers. 
}
	\label{attack_fig}
		\vspace{-.2in}
\end{figure}

\subsection{Attacks as an Optimization Problem}

The success of our proposed attack leverages core machine learning principles. When an attacker trains a sample with an incorrect label, the loss for that sample rises. However, if the same sample with its correct label is present in a benign client's training set, the loss normalizes, preserving high classification accuracy. By observing high classification accuracy on certain samples, the attacker can deduce that these samples are part of the benign clients’ training data, thus indirectly stealing data.
However, sending gradients from mislabeled samples alone would result in high angular deviation, making them easily detectable by Byzantine-robust mechanisms. To avoid detection, the attacker incorporates correctly labeled gradients to mask the malicious ones, ensuring the final crafted gradient does not appear as an outlier.
In practice, the attacker possesses a dataset containing an attack set \( D_{\text{attack}} \) (samples targeted for inference) and a mask set \( D_{\text{mask}} \) (samples for gradient masking). Using these subsets, the attacker generates malicious updates \( \bm{g}_{\text{attack}} \) from \( D_{\text{attack}} \) and mask updates \( \bm{g}_{\text{mask}} \) from \(\hat{D}_{\text{mask}} \), a subset of \( D_{\text{mask}} \). The final malicious gradient \( \bm{g}_{\text{malicious}} \) sent to the server is computed as:
\begin{align}
\bm{g}_{\text{malicious}} = \alpha \bm{g}_{\text{attack}} + \bm{g}_{\text{mask}},
\end{align}
where \( \alpha \) is a scaling factor. The attacker’s objective is to maximize the loss on \( D_{\text{attack}} \) while remaining undetected by leveraging \( \hat{D}_{\text{mask}} \). This goal can be formulated as:
\begin{align}
\label{opt_obj}
\arg\max_{\alpha, \hat{D}_{\text{mask}}} F_r\left( \mathcal{A}\left(\{\bm{g}_j\}_{j \in \mathcal{M}} \cup \{\bm{g}_i\}_{i \in \mathcal{B}} \right), D_{\text{attack}} \cup \hat{D}_{\text{mask}} \right),
\end{align}
where \( \mathcal{M} \) represents the set of malicious clients, \( \mathcal{B} \) represents the benign clients, \( r \in \mathcal{M} \) denotes a malicious client, and \( F_r \) is the local objective for client \( r \).

An illustration of our attack is shown in Figure~\ref{attack_fig}.

\subsection{Approximating the Optimization Problem}
However, directly solving the problem in Eq.~(\ref{opt_obj}) presents challenges due to the non-differentiable nature of the aggregation rule \( \mathcal{A} \). To address this, we demonstrate an approach to approximate the optimization problem. In the following sections, we provide a detailed breakdown of the steps necessary to conduct a MIA in the training phase. This outline highlights each phase’s critical components, guiding the effective execution of MIA within FL.
During the training process, the attacker leverages both \( D_{\text{attack}} \) and \( D_{\text{mask}} \) to craft malicious gradients that influence the target samples while evading detection and filtering by the server. The detailed steps involved in this process are as follows:

\myparatight{Step I. Generate attack gradients by the attack samples}To initiate the attack, the attacker modifies the original attack set \( D_{\text{attack}} \) to a new set \( \hat{D}_{\text{attack}} \) by replacing the true labels of the samples with incorrect labels, which are randomly chosen from the remaining available labels. This manipulation results in a gradient shift for the targeted samples, enabling the attacker to track variations in the loss function associated with these samples. By analyzing these loss changes, the attacker can infer whether the targeted samples are included in the training sets of benign clients, indirectly revealing private information.

\myparatight{Step II. Select mask samples based on greedy selection algorithm}To bypass the server's Byzantine-robust mechanism, a masking gradient \( \bm{g}_{\text{mask}} \) is introduced to obscure the attack gradient \( \bm{g}_{\text{attack}} \), ensuring it blends in and does not stand out as an outlier relative to the normal gradients \( \cup \bm{g}^{i}_{\text{benign}} \) produced by all benign clients, where $\bm{g}^i_{\text{benign}}$ denotes the gradient of benign client $i$.
The masking gradient \( \bm{g}_{\text{mask}} \) is created using samples from a carefully selected mask set \( \hat{D}_{\text{mask}} \) from \( D_{\text{mask}} \). 
Consequently, the attacker’s objective is to identify a subset \( \hat{D}_{\text{mask}} \), containing a fixed number of masking samples, that meets the following: 
\begin{equation}
\label{angle_constraint}
\begin{aligned}
    & \argmax_{\hat{D}_{\text{mask}} \subset {D_{\text{mask}}}} \max_{i \in \mathcal{B}} \{\angle(\bm{g}_{\text{malicious}}, \bm{g}^i_{\text{benign}}) \}\\
    & s.t. \max_{i \in \mathcal{B}} \angle(\bm{g}_{\text{malicious}}, \bm{g}^i_{\text{benign}}) \\
    & \leq \max_{i,j\in \mathcal{B}} \angle(\bm{g}^i_{\text{benign}}, \bm{g}^j_{\text{benign}}) \\
    & |\hat{D}_{\text{mask}}| = \lfloor \gamma |D_{\text{mask}}| \rfloor  \\
    & \bm{g}_{\text{malicious}} = \alpha \bm{g}_{\text{attack}} + \bm{g}_{\text{mask}},
\end{aligned}
\end{equation}
where \( \angle(\cdot) \) denotes the angle between two gradients, and \( \gamma \in (0,1) \) represents the proportion of the number of mask samples to be selected.
Solving this optimization problem—specifically, identifying the ideal mask set \( \hat{D}_{\text{mask}} \) from the pool \( D_{\text{mask}} \)—is an NP-hard challenge \cite{kempe2003maximizing}. There are \( {\binom{|D_{\text{mask}}|}{|\hat{D}_{\text{mask}}|}} \) possible combinations to evaluate for an optimal solution, making it computationally prohibitive for the attacker to exhaustively examine each alternative. To address this, we employ a greedy selection algorithm that approximates a solution to Eq.~(\ref{angle_constraint}) efficiently.
Specifically, we initialize \( \hat{D}_{\text{mask}} = \emptyset \) and iteratively add to \( \hat{D}_{\text{mask}} \) the sample in the set $D_{\text{mask}}$ \textbackslash $\hat{D}_{\text{mask}}$ that maximizes the objective function in Eq.~(\ref{angle_constraint}) when combined with the current \( \hat{D}_{\text{mask}} \). This process is repeated until the number of mask samples in \( \hat{D}_{\text{mask}} \) reaches \( \lfloor \gamma |D_{\text{mask}}| \rfloor \). The pseudocode of our greedy mask sample selection algorithm is shown in Algorithm~\ref{greedy_mask_sample_slection_algorithm}.

\begin{algorithm}[!t]
	\caption{{Greedy mask sample selection.}}
        \label{greedy_mask_sample_slection_algorithm}
	\begin{algorithmic}[1]
		\renewcommand{\algorithmicrequire}{\textbf{Input:}}
		\renewcommand{\algorithmicensure}{\textbf{Output:}}
            \Require Mask set $D_{\text{mask}}$, parameter $\gamma$.
		\Ensure  Selected mask set $\hat{D}_{\text{mask}}$.
		\State Initialize $ |\hat{D}_{\text{mask}}| = \emptyset $.
		\While {$|\hat{D}_{\text{mask}}| < \lfloor \gamma \vert D_{\text{mask}} \vert \rfloor$}

            \State Select $s = \displaystyle\argmax_{\{k\} \cup \hat{D}_{\text{mask}}} \displaystyle\max_{i \in \mathcal{B}} \{\angle(\bm{g}_{\text{malicious}}, \bm{g}^i_{\text{benign}}) \}$  
                        
            \State \hspace{1em} with $k \in  D_{\text{mask}} \setminus \hat{D}_{\text{mask}}$ and the same constraint 
            \State \hspace{1em} in Eq.~(\ref{angle_constraint}).
		  \State $\hat{D}_{\text{mask}} \leftarrow \hat{D}_{\text{mask}} \cup \{s\}$.
        \EndWhile \\
		\Return $\hat{D}_{\text{mask}}$.
	\end{algorithmic} 
\end{algorithm}

\myparatight{Step III. Optimize the scaling coefficient $\alpha$}In Step II, we begin by setting the scaling coefficient \( \alpha \) as a constant and proceed to optimize the mask set \( \hat{D}_{\text{mask}} \). Following this, we shift focus to adjusting the scaling factor \( \alpha \) while keeping the selected mask set \(\hat{D}_{\text{mask}} \) unchanged. In other words, this step involves solving the following optimization problem:
\begin{equation}
\label{alpha_optimization}
\begin{aligned}
    & \argmax_{\alpha} \max_{i \in \mathcal{B}} \{\angle(\bm{g}_{\text{malicious}}, \bm{g}^i_{\text{benign}}) \}\\
    & s.t. \max_{i \in \mathcal{B}} \angle(\bm{g}_{\text{malicious}}, \bm{g}^i_{\text{benign}}) \\
    & \leq \max_{i,j\in \mathcal{B}} \angle(\bm{g}^i_{\text{benign}}, \bm{g}^j_{\text{benign}}) \\
    & \bm{g}_{\text{malicious}} = \alpha \bm{g}_{\text{attack}} + \bm{g}_{\text{mask}}.
\end{aligned}
\end{equation}

\myparatight{Step IV. Send malicious gradients to the server}Following the completion of these steps, the attacker finalizes and transmits the crafted malicious gradient \( \bm{g}_{\text{malicious}} \) to the server. This gradient is strategically designed to evade detection, blending in with the benign gradients while carrying out the intended attack objectives.

\myparatight{Remark}%
We approximately solve the original NP-hard optimization problem using a heuristic approach. Since our attack subtly manipulates gradient directions to evade detection while achieving its goal, quantifying errors from the greedy solver is challenging. Ultimately, the attack’s real-world impact is the primary concern.

\section{Our Defense}

In this section, we introduce a defense mechanism aimed at mitigating the effects of various MIAs.
Inspired by the Trimmed-mean approach~\cite{yin2018byzantine}, our defense strategy is built on the core concept of discarding gradients identified as malicious. By filtering out these harmful gradients, our mechanism effectively strengthens the server's resilience against potential attacks, thereby enhancing the overall robustness of the FL system against privacy threats.
The pseudocode of our defense is shown in Algorithm~\ref{ATM} in Appendix.

\subsection{Motivation}

Motivated by our new attack insights and the Trimmed-mean approach, we propose Angular Trimmed-mean (\alg), a defense that leverages gradient angles to accurately detect and filter malicious updates based on directional alignment with benign gradients.
The central principle of the \alg method is to filter out gradients that exhibit directional inconsistencies, identifying them as outliers. To determine whether a gradient qualifies as an outlier, we compute the average angle between each gradient and all other gradients. This average angle serves as a basis for evaluating the gradient’s alignment with the majority. In the following section, we outline the detailed steps of our algorithm to implement this process effectively.

\myparatight{Step I}Compute the angle \( \theta_{i, j} \) between each pair of gradients \( \bm{g}_i \) and \( \bm{g}_j \) within the set \( \mathcal{G} \), where \( \mathcal{G} \) denotes the set of all benign and malicious gradients, and \( 1 \le i < j \le |\mathcal{G}| \).

\myparatight{Step II}For each gradient \( \bm{g}_k \), calculate the mean angle between \( \bm{g}_k \) and all other gradients as the following:
\begin{equation}
\label{compute average angle}
\begin{aligned}
     \bar{\theta}_k = \frac{1}{|\mathcal{G}|-1} \sum \theta_{k,l}, 
  \quad   \text{s.t. } \medspace  1 \le l \le |\mathcal{G}|, l \neq k. 
\end{aligned}
\end{equation}

\myparatight{Step III}Arrange each gradient \( \bm{g}_k \in \mathcal{G} \) in ascending order based on its mean angle. Then, construct the set \( \mathcal{G}' \) by retaining the gradients with the smallest absolute values of mean angles while removing the top \( 2b \) gradients exhibiting the largest absolute values of mean angles, where $b$ is the trim parameter and $b > 0$.

\myparatight{Step IV}Calculate the aggregated gradient \( \bar{g} \) by taking the average of the gradients selected in the set \( \mathcal{G}' \): 
\begin{equation}
\label{compute average angle}
\begin{aligned}
     \bar{\bm{g}} = \frac{1}{|\mathcal{G}|-2b} \sum_{\bm{g} \in \mathcal{G}'} \bm{g}, 
 \quad   \text{s.t. } \medspace   2b < |\mathcal{G}|.
\end{aligned}
\end{equation}

\subsection{Statistical Convergence Guarantees}
The following theorem establishes that the \alg guarantees a strictly bounded \( \ell_2 \)-deviation between post-aggregation angular measurements and their theoretical optimal values. Complete proof is provided in the Appendix~\ref{proof}.

\begin{thm}
Consider \(n\) independent and identically distributed (i.i.d.) random angles \(\{\theta_i\}_{i=1}^n\) sorted in ascending order, each drawn from a distribution \(\Omega\) with mean \(\mathbb{E}[\Omega] = \omega\) and variance \(\mathrm{Var}(\{\theta_i\}_{i=1}^n) = \sigma^2\). 
Let $b$ be the trim parameter and \(\mathcal{G}'\) represent the set of gradients remaining after applying ATM. 
If there are \(m\) malicious angles and \(2m < n\), then:
\[
\mathbb{E}\left\| \frac{1}{|\mathcal{G}'|} \sum_{\bm{g} \in \mathcal{G}'} \theta_{\bm{g}} - \omega \right\|_2^{2} \leq \frac{2(n-m)(b+1)\sigma^2}{(n-b-m)^{2}}.
\]
\label{bound}
\end{thm}


\section{Experimental Evaluation}  \label{sec:exp}

\begin{table*}[htbp]
  \centering
  \footnotesize
  \addtolength{\tabcolsep}{-1.6pt}
  \caption{Results of attack accuracy for \( C=0.8 \), where $C$ represents the proportion of clients selected in each round.}
  \label{main_results}
  \renewcommand{\arraystretch}{1.2} 
  \begin{tabular}{|c|c|cc|cc|cc|cc|cc|cc|}
    \hline
    \multirow{2}{*}{Dataset} & \multirow{2}{*}{Attack} & \multicolumn{2}{c|}{DP} & \multicolumn{2}{c|}{Top-$k$} & \multicolumn{2}{c|}{FedAvg} & \multicolumn{2}{c|}{Median} & \multicolumn{2}{c|}{Trimmed-mean} & \multicolumn{2}{c|}{\alg} \\
    \cline{3-14}
          &       & IID   & Non-IID & IID   & Non-IID & IID   & Non-IID & IID   & Non-IID & IID   & Non-IID & IID   & Non-IID \\
    \hline
    \multirow{5}{*}{Texas100} & Passive &   0.650 & 0.626 & 0.646 & 0.623 & 0.643 & 0.616 & 0.583 & 0.566 & 0.630 & 0.606 & 0.600 & 0.596 \\
          & GA &  0.826 & 0.750 & 0.826 & 0.746 & 0.826 & 0.786 & 0.820 & 0.770 & 0.810 & 0.783 & 0.766 & 0.773 \\
          & AGREvader & 0.766 & 0.756 & 0.766 & 0.767 & 0.804 & 0.840 & 0.761 & 0.750 & 0.756 & 0.767 & 0.741 & 0.754 \\
          & Adaptive &  0.749 & 0.767 & 0.757 & 0.776 & 0.784 & 0.787 & 0.761 & 0.753 & 0.766 & 0.777 & 0.743 & 0.764 \\
          & \cellcolor{greyL}FedPoisonMIA &  \cellcolor{greyL}0.890 & \cellcolor{greyL}0.906 & \cellcolor{greyL}0.891 & \cellcolor{greyL}0.893 & \cellcolor{greyL}0.897 & \cellcolor{greyL}0.897 & \cellcolor{greyL}0.913 & \cellcolor{greyL}0.930 & \cellcolor{greyL}0.880 & \cellcolor{greyL}0.887 & \cellcolor{greyL}0.803 & \cellcolor{greyL}0.853 \\
    \hline
    \multirow{5}{*}{CIFAR-10} & Passive & 0.587 & 0.580 & 0.570 & 0.570 & 0.580 & 0.570 & 0.560 & 0.573 & 0.560 & 0.560 & 0.553 & 0.590 \\
          & GA &   0.697 & 0.640 & 0.713 & 0.613 & 0.727 & 0.610 & 0.720 & 0.597 & 0.727 & 0.657 & 0.651 & 0.623 \\
          & AGREvader &  0.651 & 0.677 & 0.643 & 0.670 & 0.700 & 0.613 & 0.703 & 0.590 & 0.703 & 0.637 & 0.627 & 0.643 \\
          & Adaptive &  0.658 & 0.676 & 0.657 & 0.672 & 0.681 & 0.697 & 0.668 & 0.614 & 0.687 & 0.668 & 0.650 & 0.643 \\
          & \cellcolor{greyL}FedPoisonMIA & \cellcolor{greyL}0.913 & \cellcolor{greyL}0.777 & \cellcolor{greyL}0.927 & \cellcolor{greyL}0.767 & \cellcolor{greyL}0.913 & \cellcolor{greyL}0.786 & \cellcolor{greyL}0.753 & \cellcolor{greyL}0.660 & \cellcolor{greyL}0.857 & \cellcolor{greyL}0.803 & \cellcolor{greyL}0.713 & \cellcolor{greyL}0.650 \\
    \hline
    \multirow{5}{*}{STL10} & Passive &  0.603 & 0.603 & 0.610 & 0.583 & 0.617 & 0.593 & 0.573 & 0.547 & 0.587 & 0.573 & 0.613 & 0.553 \\
          & GA &  0.820 & 0.713 & 0.830 & 0.733 & 0.807 & 0.723 & 0.786 & 0.650 & 0.800 & 0.693 & 0.730 & 0.677 \\
          & AGREvader &  0.730 & 0.727 & 0.691 & 0.753 & 0.797 & 0.733 & 0.776 & 0.710 & 0.810 & 0.747 & 0.687 & 0.670 \\
          & Adaptive &  0.714 & 0.750 & 0.724 & 0.753 & 0.703 & 0.733 & 0.695 & 0.760 & 0.707 & 0.753 & 0.694 & 0.733 \\
          & \cellcolor{greyL}FedPoisonMIA & \cellcolor{greyL}0.920 & \cellcolor{greyL}0.837 & \cellcolor{greyL}0.917 & \cellcolor{greyL}0.827 & \cellcolor{greyL}0.900 & \cellcolor{greyL}0.817 & \cellcolor{greyL}0.863 & \cellcolor{greyL}0.783 & \cellcolor{greyL}0.833 & \cellcolor{greyL}0.847 & \cellcolor{greyL}0.827 & \cellcolor{greyL}0.773 \\
    \hline
    \multirow{5}{*}{FER2013} & Passive & 0.640 & 0.596 & 0.630 & 0.626 & 0.630 & 0.580 & 0.586 & 0.536 & 0.603 & 0.566 & 0.626 & 0.686 \\
          & GA &  0.803 & 0.800 & 0.776 & 0.796 & 0.823 & 0.766 & 0.736 & 0.683 & 0.790 & 0.790 & 0.733 & 0.713 \\
          & AGREvader &  0.840 & 0.763 & 0.763 & 0.780 & 0.883 & 0.853 & 0.773 & 0.746 & 0.744 & 0.736 & 0.752 & 0.743 \\
          & Adaptive &  0.836 & 0.843 & 0.774 & 0.803 & 0.854 & 0.863 & 0.779 & 0.749 & 0.804 & 0.816 & 0.741 & 0.756 \\
          & \cellcolor{greyL}FedPoisonMIA & \cellcolor{greyL}0.920 & \cellcolor{greyL}0.910 & \cellcolor{greyL}0.926 & \cellcolor{greyL}0.960 & \cellcolor{greyL}0.933 & \cellcolor{greyL}0.953 & \cellcolor{greyL}0.816 & \cellcolor{greyL}0.766 & \cellcolor{greyL}0.936 & \cellcolor{greyL}0.913 & \cellcolor{greyL}0.786 & \cellcolor{greyL}0.806 \\
    \hline

    \end{tabular}%
\vspace{-.1in}
\end{table*}

\subsection{Experimental Setup}

\subsubsection{Datasets}


We assess our proposed attack, defense mechanism, and baseline methods using four real-world datasets: CIFAR-10~\cite{krizhevsky2009learning}, STL10~\cite{coates2011analysis}, Texas100~\cite{texas100}, and FER2013~\cite{goodfellow2013challenges}. 
See Appendix~\ref{data_app} for details.

\subsubsection{Comparison PMIA}
In our experiment, we employ several baseline PMIA methods to assess the effectiveness of our proposed attack: Passive Membership Inference Attack~\cite{leino2020stolen}, Gradient Ascent (GA)~\cite{nasr2019comprehensive}, AGREvader~\cite{zhang2023agrevader}, and Adaptive attack. A complete attack description is presented in Appendix~\ref{attack_decp}.

\subsubsection{Comparison Defenses}

We evaluate the performance of our attack and defense using various typical robust mechanisms: FedAvg~\cite{McMahan17}, Median~\cite{yin2018byzantine}, Trimmed-mean~\cite{yin2018byzantine}, Differential Privacy~\cite{dwork2008differential}, Top-$k$~\cite{aji2017sparse}, Multi-Krum~\cite{blanchard2017machine}, Fang~\cite{fang2020local}, and DeepSight~\cite{rieger2022deepsight}. Details of these mechanisms are listed in Appendix~\ref{aggregate_rule}.

\begin{table*}[htbp]
  \centering
  \footnotesize
  \addtolength{\tabcolsep}{-1.6pt}
  \caption{Results of attack accuracy for partial-knowledge attack.}
  \label{partial_attack}
  \renewcommand{\arraystretch}{1.2} 
  \begin{tabular}{|c|c|cc|cc|cc|cc|cc|cc|}
    \hline
    \multirow{2}{*}{Dataset} & \multirow{2}{*}{Attack} & \multicolumn{2}{c|}{DP} & \multicolumn{2}{c|}{Top-$k$} & \multicolumn{2}{c|}{FedAvg} & \multicolumn{2}{c|}{Median} & \multicolumn{2}{c|}{Trimmed-mean} & \multicolumn{2}{c|}{\alg} \\
    \cline{3-14}
          &       & IID   & Non-IID & IID   & Non-IID & IID   & Non-IID & IID   & Non-IID & IID   & Non-IID & IID   & Non-IID \\
    \hline
    \multirow{5}{*}{Texas100} 
          & Passive        & 0.583 & 0.610 & 0.583 & 0.603 & 0.587 & 0.603 & 0.533 & 0.560 & 0.573 & 0.597 & 0.577 & 0.600 \\
          & GA             & 0.783 & 0.763 & 0.780 & 0.753 & 0.773 & 0.747 & 0.740 & 0.673 & 0.760 & 0.733 & 0.720 & 0.717 \\
          & AGREvader      & 0.730 & 0.717 & 0.730 & 0.707 & 0.767 & 0.747 & 0.753 & 0.683 & 0.767 & 0.747 & 0.717 & 0.700 \\
          & Adaptive       & 0.728 & 0.730 & 0.721 & 0.720 & 0.717 & 0.733 & 0.723 & 0.700 & 0.730 & 0.713 & 0.721 & 0.727 \\
          & \cellcolor{greyL}FedPoisonMIA & \cellcolor{greyL}0.853 & \cellcolor{greyL}0.873 & \cellcolor{greyL}0.853 & \cellcolor{greyL}0.860 & \cellcolor{greyL}0.843 & \cellcolor{greyL}0.857 & \cellcolor{greyL}0.829 & \cellcolor{greyL}0.838 & \cellcolor{greyL}0.823 & \cellcolor{greyL}0.807 & \cellcolor{greyL}0.757 & \cellcolor{greyL}0.767 \\
    \hline
  \end{tabular}%
\vspace{-.1in}
\end{table*}

\begin{table*}[!t]
  \centering
    \footnotesize
  \addtolength{\tabcolsep}{-1.6pt}
  \renewcommand{\arraystretch}{1.2}
  \caption{Comparison of attack accuracy between random selection and greedy selection.}
    \begin{tabular}{|c|c|c|c|c|c|c|c|c|c|c|c|c|}
    \hline
    \multirow{2}{*}{Method} & \multicolumn{2}{c|}{DP} & \multicolumn{2}{c|}{Top-$k$} & \multicolumn{2}{c|}{FedAvg} & \multicolumn{2}{c|}{Median} & \multicolumn{2}{c|}{Trimmed-mean} & \multicolumn{2}{c|}{\alg} \\
\cline{2-13}          & IID  &Non-IID   & IID  & Non-IID   & IID    & Non-IID   &IID   & Non-IID   & IID   & Non-IID   & IID&Non-IID  \\
    \hline
    Random & 0.748 & 0.757 & 0.751 & 0.757 & 0.771 & 0.761 & 0.744 & 0.767 & 0.751 & 0.748 & 0.751 & 0.774 \\
    \hline
    Greedy & 0.887 & 0.894 & 0.914 & 0.947 & 0.884 & 0.894 & 0.807 & 0.837 & 0.884 & 0.890 & 0.880 & 0.904 \\
    \hline
    \end{tabular}%
  \label{cover_selectionl}%
\vspace{-.1in}
\end{table*}%

\subsubsection{Synchronous and Asynchronous Setting}

We assess the performance of our methods in both synchronous and asynchronous scenarios. In the synchronous setting, the server updates the global model only after receiving updates from all clients. On the other hand, in the asynchronous setting, the server immediately updates the global model upon receiving a single client's model update, without waiting for the others. To simulate asynchronous behavior, we follow the approach outlined in~\cite{fang2022aflguard}, randomly sampling client delays from the interval \([0, \tau_\text{{max}}]\), where \(\tau_\text{{max}}\) is set to 5 by default.

\subsubsection{Parameters Setting}
The default parameter settings for the FL setup, the composition of the attack and mask sets, as well as the model and training details, are provided in Appendix~\ref{parameter_setting}.

\subsubsection{Non-IID Setting}

We evaluate both independent and identically distributed (IID) and Non-IID settings using four real-world datasets. To simulate the Non-IID setting, we employ a group-based data partitioning strategy~\cite{fang2020local}. Specifically, we divide the clients into \(h\) groups, with each group corresponding to one of the dataset’s class. A sample with label \(q\) is assigned to the \(q\)-th group with a probability bias \(\beta\), while the remaining groups receive the sample with a probability of \(\frac{1-\beta}{h-1}\). Each client within a group receives training examples in a balanced manner. By default, we set \(\beta = 0.5\).

\subsubsection{Evaluation Metrics}

In line with prior research~\cite{song2019membership, shokri2017membership, rezaei2021difficulty}, we consider the following three key metrics:

\myparatight{\textbf{Attack accuracy}}Attack accuracy is the highest proportion of correctly identified samples in the best-performing round of the attack process.

\myparatight{\textbf{Attack precision}}Attack precision evaluates the ratio of correctly predicted true members to the total number of predicted true and false members.

\myparatight{\textbf{Attack recall}}Attack recall measures the fraction of correctly predicted true members among all actual members. It reflects the attack's effectiveness in identifying all member samples, highlighting the method's ability to capture as many targets as possible.

\begin{figure*}[!t]
	\centering
	\includegraphics[scale = 0.46]{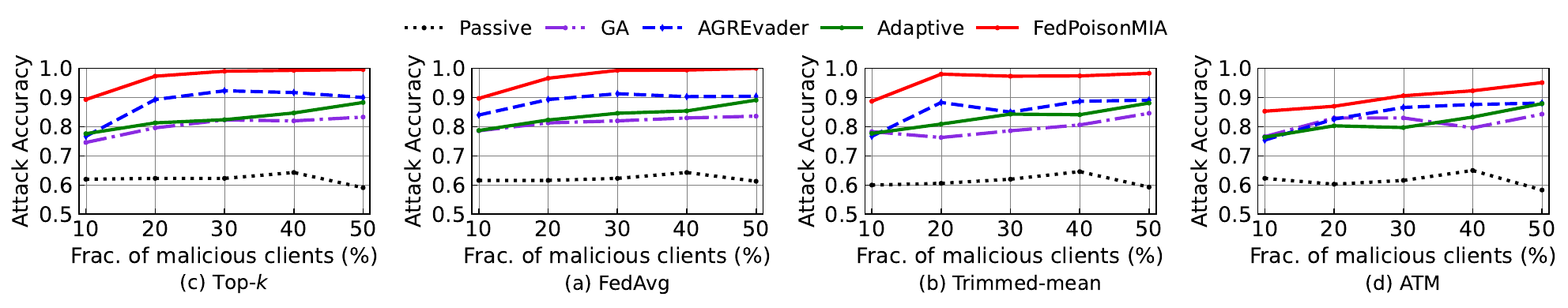}
	\caption{Impact of fraction of malicious clients. 
}
	\label{num_malicious_clients}
		\vspace{-.15in}
\end{figure*}

\begin{figure*}[!t]
	\centering
	\includegraphics[scale = 0.46]{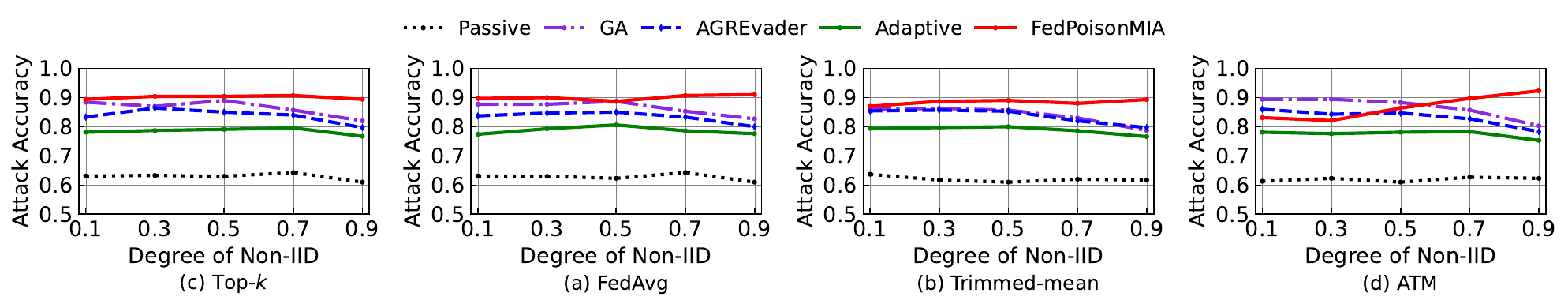}
	\caption{Impact of degree of Non-IID. 
}
	\label{non_iid}
		\vspace{-.15in}
\end{figure*}

\subsection{Experimental Results}

\myparatight{Our proposed attack is effective}%
We report results on four real-world datasets. The attack accuracies of different methods under various defense mechanisms are summarized in Table~\ref{main_results} (for $C = 0.8$) and Table~\ref{attack_acc_syn_1.0} (for $C = 1.0$) in Appendix, where $C$ denotes the fraction of clients selected per round.
Our proposed attack consistently achieves the highest accuracy across all datasets and scenarios, demonstrating its superiority over three baseline methods. For example, on the Texas100 dataset with \(C = 0.8\) under a Non-IID distribution, our attack improves attack accuracy by 15.6\%, 11.6\%, 5.7\%, 16.0\%, 10.3\%, and 8.0\% compared to the best baseline attacks under different defense mechanisms.
The results also show higher attack accuracy in IID settings compared to Non-IID, indicating that the consistency of updates in IID data facilitates the attack. For instance, in the IID setting with \(C=0.8\), our attack achieves an accuracy of 86.3\% on the STL10 dataset with Median defense, compared to 78.3\% in Non-IID settings. Additionally, the results for \(C=0.8\) and \(C=1.0\) are similar, suggesting that increasing client participation does not significantly affect the attack’s effectiveness.

We further evaluate the attack in an asynchronous FL setting, with attack accuracies for both $C = 0.8$ and $C = 1.0$ reported in Table~\ref{asyn_attack_acc_result} (Appendix). In addition to attack accuracy, we consider attack precision and attack recall as key metrics for assessing attack performance. The attack precision results under the synchronous setting for both $C = 0.8$ and $C = 1.0$ are presented in Table~\ref{precisoion_result} (Appendix), while the asynchronous results are shown in Table~\ref{precisoion_result_2}. Attack recall under synchronous and asynchronous settings is provided in Table~\ref{recall_result} and Table~\ref{recall_result_2}, respectively. The global model’s test accuracies under both settings are reported in Table~\ref{test_acc_result} and Table~\ref{test_acc_result_2}. Additional evaluations against Byzantine-robust defenses such as Fang, Multi-Krum, and DeepSight on the Texas100 dataset are shown in Table~\ref{addition_defense}.

\myparatight{Partial-knowledge attack}%
Additionally, we evaluate the performance of our attack under partial-knowledge attack scenario, with the detailed results presented in Table~\ref{partial_attack}. In partial-knowledge attack, our attack consistently surpasses other baseline methods across various defense mechanisms. Remarkably, the efficacy of our attack demonstrates consistent robustness, with a maximum performance degradation of merely 8.0\% across all defense mechanisms under both IID and Non-IID settings, when compared to the full-knowledge attack scenario.

\myparatight{Our greedy mask sample selection algorithm is effective}Table~\ref{cover_selectionl} presents the attack accuracy results on the Texas100 dataset, comparing the mask samples selected by our proposed greedy selection algorithm with those chosen randomly. Our approach notably surpasses random selection, showing an improvement of over 20\% in attack accuracy. This demonstrates that the carefully selected mask samples play a crucial role in helping the malicious gradient evade detection by existing Byzantine-robust mechanisms.

\myparatight{Our proposed defense is effective}As shown in Table~\ref{main_results}, Table~\ref{partial_attack}, and Table~\ref{addition_defense}, \alg achieves consistently lower attack accuracy across most settings compared to other defense methods, showcasing its effectiveness and robustness, especially when defending against our proposed attack method.
For instance, on the Texas100 dataset with IID distribution under $C=0.8$, \alg restricts the accuracy of our attack to 80.3\%, significantly lower than other defenses, such as Median (89.13\%) and Top-$k$ (89.1\%).

Moreover, our \alg mechanism demonstrates strong resilience against adaptive attacks, consistently maintaining the lowest attack accuracy among all baseline defenses. Notably, it also effectively suppresses the attack accuracy to a level lower than that achieved by our proposed attack method.
Table~\ref{test_acc_result} in Appendix demonstrates that our defense achieves strong defense performance without compromising the training effectiveness of the global model, maintaining similar test accuracy to the scenario without malicious clients in FL.
Furthermore, \alg does not impose any additional computational cost on the clients.

\begin{table}[t]
  \centering
    \footnotesize
  \addtolength{\tabcolsep}{-2.5pt}
  \caption{Results of attack accuracy on Texas100 under additional defense mechanisms.}
    \begin{tabular}{|c|c|c|c|c|c|c|}
      \hline
    \multirow{2}{*}{Attack} & \multicolumn{2}{c|}{Fang} & \multicolumn{2}{c|}{Multi-Krum} & \multicolumn{2}{c|}{DeepSight} \\
\cline{2-7}          & IID   & Non-IID & IID   & Non-IID & IID   & Non-IID \\
     \hline
    Passive & 0.643 & 0.643 & 0.606 & 0.603 & 0.730 & 0.706 \\
     \hline
    GA    & 0.833 & 0.806 & 0.863 & 0.830 & 0.826 & 0.793 \\
     \hline
    AGREvader & 0.903 & 0.910 & 0.836 & 0.876 & 0.810 & 0.846 \\
    \hline
    Adaptive   &  0.738 & 0.756 & 0.693 & 0.723 & 0.693 & 0.733 \\
     \hline
    \cellcolor{greyL}FedPoisonMIA   &  \cellcolor{greyL}0.930 & \cellcolor{greyL}0.950 & \cellcolor{greyL}0.890 & \cellcolor{greyL}0.906 & \cellcolor{greyL}0.840 & \cellcolor{greyL}0.853 \\
     \hline
     
    \end{tabular}%
  \label{addition_defense}%
	\vspace{-.2in}
\end{table}%

\myparatight{Impact of fraction of malicious clients}%
Fig.~\ref{num_malicious_clients} presents the attack results on the Texas100 dataset as the proportion of malicious clients increases from 10\% to 50\%, with the total number of clients fixed at 10. 
Note that in Fig.~\ref{num_malicious_clients}, we compare \alg only with Top-$k$, FedAvg, and Trimmed-mean, as Top-$k$ represents a typical non-aggregation-based defense, FedAvg serves as the standard aggregation rule in FL, and Trimmed-mean is a representative aggregation-based defense.
We observe that attack accuracy rises as the number of malicious clients increases. Notably, our attack method achieves nearly 100\% attack accuracy when malicious clients constitute 30\%, under Top-$k$, FedAvg, and Trimmed-mean, and while baseline attacks hover around 90\%. However, under our proposed defense, the attack accuracy is significantly reduced, highlighting the effectiveness and robustness of our defense mechanism.

\myparatight{Impact of degree of Non-IID}%
Fig.~\ref{non_iid} illustrates the impact of varying the Non-IID degree, controlled by parameter $\beta$, on attack accuracy under different defense mechanisms and Texas100 dataset is considered. 
The Non-IID levels are set to \( \{0.1, 0.3, 0.5, 0.7, 0.9\} \).
Our attack method consistently achieves high accuracy, even under the extreme Non-IID distribution of 90\%. In contrast, other attacks, such as Gradient Ascent and AGREvader, see a decline in accuracy as the degree of Non-IID increases.

\myparatight{Impact of the total number of clients}As shown in Fig.~\ref{num_clients} in Appendix, the attack accuracy decreases overall as the total number of clients increases in FL training on the Texas100 dataset, with client numbers varying in \( \{8, 10, 15, 20, 30\} \) while maintaining a constant number of 1 malicious client. This decline occurs because, with more clients, the proportion of the malicious gradient in the aggregated gradient becomes smaller, reducing its influence on the global model update.
Although the attack accuracy decreases, our method still achieves the highest attack accuracy compared to baseline attack methods. Additionally, all attack methods show relatively low accuracy when evaluated against our proposed defense.

\myparatight{Impact of number of attack sample}%
Fig.~\ref{target_sample} in Appendix shows the attack accuracy results for different numbers of attack samples in \(\{100, 200, 300, 400, 500\} \), where Texas100 dataset is considered. 
We observe that as the number of attack samples increases, the attack accuracy gradually decreases across various defense mechanisms. Our attack method consistently achieves the highest accuracy under FedAvg, Trimmed-mean, and Top-$k$. However, our \alg effectively reduces the attack accuracy for all methods, particularly our proposed attack, achieving lower accuracy compared to other defenses in most cases.

\myparatight{Computational overhead of our attack and defense}%
Table~\ref{attack_cost} (Appendix) presents the runtime of our proposed attack compared to AGREvader across four datasets, showing that the total execution time is comparable to AGREvader, thereby demonstrating its practicality.
Table~\ref{defense_cost} (Appendix) summarizes the runtime of each defense under settings with 10 and 50 clients. \alg incurs significantly lower computational overhead than other approaches. This efficiency stems from the fact that the dominant cost in both baselines and \alg is due to sorting. While Median and Trimmed-mean perform sorting over high-dimensional parameter vectors (typically exceeding 100,000 dimensions), \alg only requires sorting over the number of clients, which is considerably smaller.


\section{Conclusion}

We propose a poisoning membership inference attack (PMIA) that optimizes malicious gradients to maximize target update deviations while remaining indistinguishable from benign ones, evading detection. This exposes a major privacy risk in FL. To defend against PMIA, we introduce a Byzantine-robust mechanism that filters updates with significant angular deviations. Extensive experiments validate the effectiveness of both our attack and defense.

\section*{Acknowledgments}
We thank the anonymous reviewers for their comments.

{
\small
\bibliographystyle{ieeenat_fullname}
\bibliography{refs}
}


\clearpage
\appendix

\section{Convergence Analysis of Angular Trimmed-mean Aggregation (\alg)}
\label{proof}

Before proving Theorem~\ref{bound}, we first present Lemma~\ref{lemma_1}. The proof is partially inspired by~\cite{xie2018phocas}.

\begin{lem}
\label{lemma_1}
Let $\{\theta_i\}_{i=1}^n$ be a sorted sequence of scalar values in ascending order, where $m$ entries are assumed to be malicious. For clarity, we refer to the remaining $n - m$ benign values as $\{\hat{\theta}_i\}_{i=1}^{n - m}$, which form a subset of the original sequence.
Thus, for \(m < b \leq \lfloor n/2 \rfloor - 1\), 
\[
\hat{\theta}_{b-m+i} \stackrel{(I)} \leq \theta_{b+i} \stackrel{(II)} \leq \hat{\theta}_{b+i},\quad 1 \leq i \leq n - 2b,
\] 
where \(\hat{\theta}_{b+i}\) is the \((b+i)\)-th smallest element in \(\{\hat{\theta}_i\}_{i=1}^{n-m}\), and \(\theta_{b+i}\) is the \((b+i)\)-th smallest element in \(\{\theta_i\}_{i=1}^n\).
\label{lem}
\end{lem}

\begin{proof}
We prove each of the two inequalities individually.

\myparatight{Inequality I}%
Suppose for contradiction that \(\hat{\theta}_{b-m+i} > \theta_{b+i}\). This implies there exist \((n-m) - (b-m+i) + 1 = n - b - i + 1\) correct values strictly greater than \(\theta_{b+i}\). However, since \(\theta_{b+i}\) is the \((b+i)\)-th smallest element in \(\{\theta_i\}_{i=1}^n\), there can be at most \(n - (b+i) = n - b - i\) elements greater than \(\theta_{b+i}\). This contradiction establishes \(\hat{\theta}_{b-m+i} \leq \theta_{b+i}\).

\myparatight{Inequality II}%
Suppose for contradiction that \(\theta_{b+i} > \hat{\theta}_{b+i}\). This implies there exist \(b+i\) correct values strictly less than \(\theta_{b+i}\). However, since \(\theta_{b+i}\) is the \((b+i)\)-th smallest element in\(\{\theta_i\}_{i=1}^n\),  there can be at most \(b+i-1\) elements less than \(\theta_{b+i}\). This contradiction establishes \(\theta_{b+i} \leq \hat{\theta}_{b+i}\).
\end{proof}

\myparatight{Proof of Theorem~\ref{bound}}
According to Lemma~\ref{lem}, we have 
\begin{align*}
&\sum_{i=b-m+1}^{n-b-m}(\hat{\theta}_i - \omega) 
\leq \sum_{i=b+1}^{n-b}(\theta_i - \omega) 
\leq \sum_{i=b+1}^{n-b}(\hat{\theta}_i - \omega) \\
\Rightarrow \quad 
&\frac{\sum_{i=1}^{n-b-m}(\hat{\theta}_i - \omega)}{n-b-m} 
\leq \frac{\sum_{i=b+1}^{n-b}(\theta_i - \omega)}{n-2b} 
\leq \frac{\sum_{i=b+1}^{n-m}(\hat{\theta}_i - \omega)}{n-b-m} \\
\Rightarrow \quad 
&\left[\frac{\sum_{i=b+1}^{n-b}(\theta_i- \omega)}{n-2b}\right]^2 \\
&\leq \max\left\{
\left[\frac{\sum_{i=1}^{n-b-m}(\hat{\theta}_i - \omega)}{n-b-m}\right]^2,
\left[\frac{\sum_{i=b+1}^{n-m}(\hat{\theta}_i - \omega)}{n-b-m}\right]^2 
\right\}.
\end{align*}

Thus, one has that:
\begin{align*}
&\left[\,\frac{1}{|\mathcal{G}'|} \sum_{\bm{g} \in \mathcal{G}'} \theta_{\bm{g}} - \omega\right]^2 \\
&= \left[\frac{\sum_{i=b+1}^{n-b}\theta_i}{n-2b} - \omega \right]^2 \\
&= \left[\frac{\sum_{i=b+1}^{n-b}(\theta_i - \omega)}{n-2b}\right]^2 \\
&\leq \max\left\{ 
\left[\frac{\sum_{i=1}^{n-b-m}(\hat{\theta}_{i} - \omega)}{n-b-m}\right]^2, 
\left[\frac{\sum_{i=b+1}^{n-m}(\hat{\theta}_{i} - \omega)}{n-b-m}\right]^2 
\right\}.
\end{align*}

Note that for any subset $T \subseteq [n - m]$ with size $|T| = n - b - m$, the following bound holds:
\begin{align*}
    & \left[ \frac{\sum_{i \in T} (\hat{\theta}_i - \omega)}{n - b - m} \right]^2 \nonumber \\
    &= \left[ \frac{\sum_{i \in [n - m]} (\hat{\theta}_i - \omega) - \sum_{i \notin T} (\hat{\theta}_i - \omega)}{n - b - m} \right]^2 \nonumber \\
    &\leq 2 \left[ \frac{\sum_{i \in [n-m]} (\hat{\theta}_i - \omega)}{n - b - m} \right]^2 
        + 2 \left[ \frac{\sum_{i \notin T} (\hat{\theta}_i - \omega)}{n - b - m} \right]^2 \nonumber \\
    &= \frac{2(n - m)^2}{(n - b - m)^2} \left[ \frac{\sum_{i \in [n-m]} (\hat{\theta}_i - \omega)}{n - m} \right]^2  \nonumber\\
        &\quad+ \frac{2b^2}{(n - b - m)^2} \left[ \frac{\sum_{i \notin T} (\hat{\theta}_i - \omega)}{b} \right]^2 \nonumber \\
    &\leq \frac{2(n - m)^2}{(n - b - m)^2} \left[ \frac{\sum_{i \in [n-m]} (\hat{\theta}_i - \omega)}{n - m} \right]^2 \nonumber \\
        &\quad + \frac{2b^2}{(n - b - m)^2} \frac{\left[\sum_{i \notin T} (\hat{\theta}_i - \omega) \right]^2}{b} \nonumber \\
    &\leq \frac{2(n - m)^2}{(n - b - m)^2} \left[ \frac{\sum_{i \in [n-m]} (\hat{\theta}_i - \omega)}{n - m} \right]^2  \nonumber \\
        &\quad+ \frac{2b^2}{(n - b - m)^2} \frac{\left[ \sum_{i \in [n-m]}  (\hat{\theta}_i - \omega) \right] ^2}{b}.
\end{align*}

Taking the expectation yields:
\begin{align*}
&\quad \mathbb{E} \left[ \frac{\sum_{i \in T} (\hat{\theta}_i - \omega)}{n - b - m} \right]^2 \\ 
&\leq \frac{2(n - m)^2}{(n - b - m)^2} \cdot \frac{\sigma^2}{n - m} + \frac{2b^2}{(n - b - m)^2} \cdot \frac{(n - m) \sigma^2}{b} \\
&= \frac{2(n - m) \sigma^2}{(n - b - m)^2} + \frac{2b(n - m) \sigma^2}{(n - b - m)^2} \\
&= \frac{2(n - m)(b + 1) \sigma^2}{(n - b - m)^2}.
\end{align*}

Putting all the above components together, one has the following:
\[
\mathbb{E}\left\| \frac{1}{|\mathcal{G}'|} \sum_{\bm{g} \in \mathcal{G}'} \theta_{\bm{g}} - \omega \right\|_2^{2} \leq \frac{2(n-m)(b+1)\sigma^2}{(n-b-m)^{2}}.
\]

The proof is complete.


\section{Dataset Description}
\label{data_app}

Detailed descriptions of the datasets used to evaluate our attack and defense method are provided below.
\myparatight{Texas100~\cite{texas100}}This dataset comprises hospital discharge records, containing inpatient data from various medical facilities, as published by the Texas Department of State Health Services. It includes 67,330 records with 6,170 binary features representing the 100 most frequently performed medical procedures. The records are organized into 100 distinct categories, each representing a unique patient type.

\myparatight{CIFAR-10~\cite{krizhevsky2009learning}}This dataset is a well-established benchmark for real-world object recognition, comprising 60,000 color images distributed evenly across 10 classes. It includes 50,000 images for training and 10,000 for testing, with a balanced number of images in each class.

\myparatight{STL10~\cite{coates2011analysis}}Like CIFAR-10, this dataset is designed for image recognition and includes 10 classes, with 5,000 labeled images for training and 8,000 images for testing.

\myparatight{FER2013~\cite{goodfellow2013challenges}}This dataset consists of 35,886 grayscale images depicting facial expressions, divided into 28,708 training images, 3,589 PublicTest images, and 3,589 PrivateTest images. The images represent seven expression categories: anger, disgust, fear, happiness, sadness, surprise, and neutral.

\section{Attack Description} 
\label{attack_decp}

\myparatight{Passive Membership Inference Attack~\cite{song2019privacy}}Once the global model is downloaded from the server, the attacker determines an input sample to be a member if the model predicts it correctly; otherwise, it is classified as a non-member.

\myparatight{Gradient Ascent (GA)~\cite{nasr2019comprehensive}}The attack uses gradient ascent on target samples to heighten the prediction gap between members and non-members. Upon receiving the global model parameters, it conducts inference in the manner of a passive membership inference attack.

\myparatight{AGREvader~\cite{zhang2023agrevader}}Rather than only altering the attack samples, AGREvader blends the attack gradients with normal gradients to ensure that the resulting combined gradients remain close to benign gradients in Euclidean norm, preventing noticeable deviation.

\myparatight{Adaptive attack}We examine a strong adversarial setting where the attacker is fully aware of the server's use of \alg. In this scenario, an adaptive attack is devised by carefully constructing gradients that inherently evade ATM filtering. The pseudocode for this attack is presented in Algorithm~\ref{alg:adaptive_attack}.

\section{Comparison Defenses} 
\label{aggregate_rule}

We evaluate the performance of our attack and defense using the following mechanisms:

\myparatight{Differential Privacy~\cite{dwork2008differential}}
The server adds Gaussian noise to all received gradients before performing the aggregation operation.

\myparatight{Top-$k$~\cite{aji2017sparse}}
This approach selects the top \( k \) gradient dimensions with the highest absolute values for updates in the aggregation process, setting all other dimensions to zero.

\myparatight{FedAvg~\cite{McMahan17}}This trivial aggregation rule takes a simple average of the client updates. 

\myparatight{Median~\cite{yin2018byzantine}}This method computes the element-wise median of the gradients in the set $\mathcal{G}$, where $\mathcal{G}$ denotes all clients' uploaded gradients.

\myparatight{Trimmed-mean~\cite{yin2018byzantine}}
Once the server receives the set of all selected update gradients $\mathcal{G}$, for each dimension, it removes the largest $b$ and smallest $b$ elements before calculating their average.

\myparatight{Multi-Krum~\cite{blanchard2017machine}}Upon receiving each model update, the server begins by identifying the \(n-f-1\) updates that are closest in terms of Euclidean distance, where $f$ is the number of malicious clients. It then computes a cumulative score by aggregating these nearby updates. The update with the lowest calculated score is subsequently added to a candidate set. This selection and scoring process is repeated iteratively until a total of \(k\) updates have been selected. Once the candidate set is complete, the server updates the global model by aggregating all the chosen candidate updates. This method ensures that only the most consistent and reliable updates contribute to the global model, enhancing the robustness and accuracy of the federated learning system.

\myparatight{Fang~\cite{fang2020local}}The Fang defense method utilizes two techniques: Error Rate Rejection (ERR) and Loss Function Rejection (LFR), to filter out gradients from potentially malicious participants. By removing gradients that most negatively affect the error rate and loss, respectively, these methods strengthen the model's robustness. This selective exclusion helps ensure that only gradients that contribute positively to the model's performance are retained, improving its overall resilience against adversarial influences.

\myparatight{DeepSight~\cite{rieger2022deepsight}}The mechanism begins by calculating division differences and normalized update energies, then clusters the update gradients based on these metrics and cosine similarity. The cluster labels are refined through a voting scheme. Afterward, \( \ell_2 \)-norm clipping is applied to each benign gradient, and the clipped gradients are aggregated to update the global model. This process ensures that only reliable gradients contribute to the model update, enhancing the system's robustness.

\section{Parameters Setting} 
\label{parameter_setting}

In the default FL training scenario, there are 10 clients in total, consisting of both benign and malicious clients, with 10\% of the clients being malicious and conducting the full-knowledge attack. In the partial-knowledge attack scenario, there are 50 clients in total, with 10 clients being malicious. During each training round, we assume that 80\% of the clients participate in the training process. The attacker possesses 300 attack samples (\( |D_{\text{attack}}| \)) and 300 masking samples (\( |D_{\text{mask}}| \)). By default, we set \( \gamma = 0.1 \) when constructing \( \hat{D}_{\text{mask}} \). To account for the worst-case scenario, we assume the attacker begins launching their attack in the first training round.
For model training, we utilized the ResNet-20~\cite{krizhevsky2009learning} architecture on the CIFAR-10, STL10, and FER2013 datasets, while employing fully connected models for the Texas100 dataset. For all datasets, we set the training duration to 800 epochs, with a batch size of 64 and a learning rate of 0.01. To optimize the model, we used the Adam optimizer~\cite{kingma2014adam}, which dynamically adjusts the learning rate, momentum, and other training parameters throughout the training process.

\begin{figure*}[!t]
	\centering
	\includegraphics[scale = 0.46]{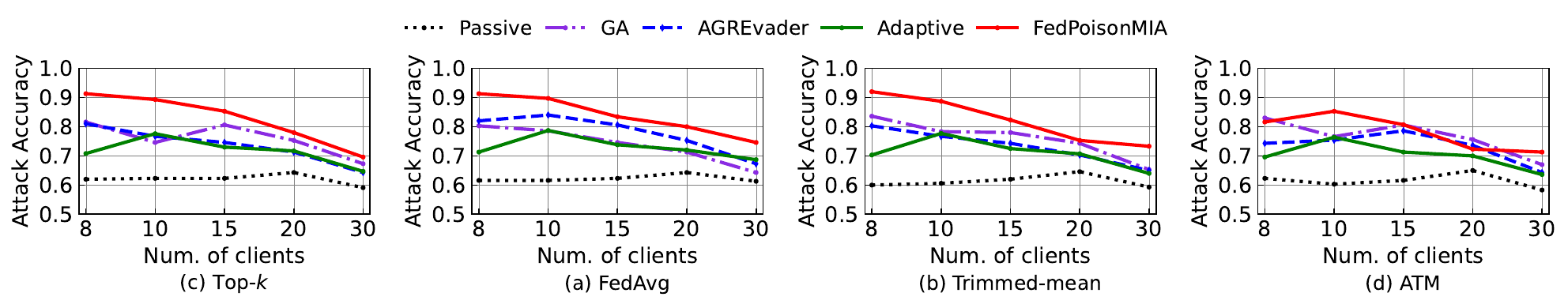}
	\caption{Impact of total number of clients. 
}
	\label{num_clients}
		\vspace{-.1in}
\end{figure*}

\begin{figure*}[htbp]
	\centering
	\includegraphics[scale = 0.46]{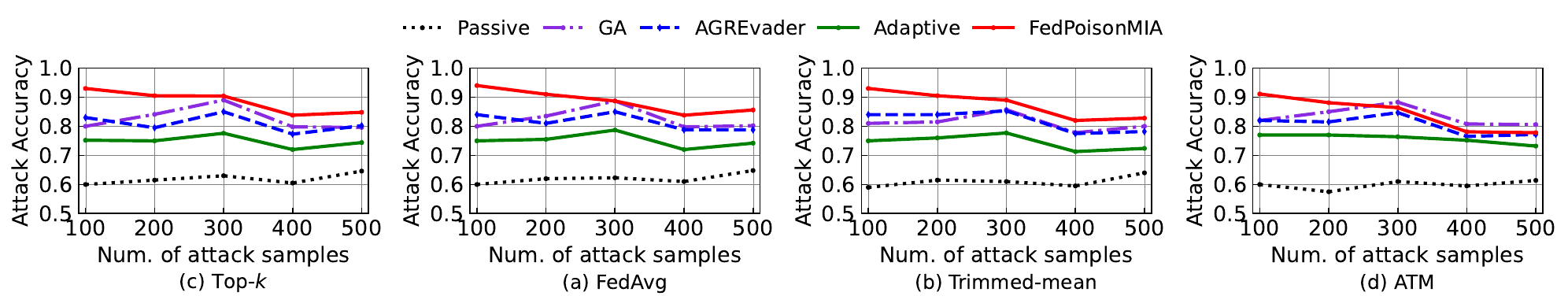}
	\caption{Impact of number of target samples. 
}
	\label{target_sample}
		\vspace{-.1in}
\end{figure*}

\begin{table*}[htbp]
  \centering
    \footnotesize
  \addtolength{\tabcolsep}{-1.6pt}
  \renewcommand{\arraystretch}{1.2}
  \caption{Attack accuracy with \( C=1.0 \) in synchronous setting, where $C$ represents the proportion of clients selected in each round.}
  \label{attack_acc_syn_1.0}
  \renewcommand{\arraystretch}{1.2} 
\begin{tabular}{|c|c|cc|cc|cc|cc|cc|cc|}
    \hline
    \multirow{2}{*}{Dataset} & \multirow{2}{*}{Attack} & \multicolumn{2}{c|}{DP} & \multicolumn{2}{c|}{Top-$k$} & \multicolumn{2}{c|}{FedAvg} & \multicolumn{2}{c|}{Median} & \multicolumn{2}{c|}{Trimmed-Mean} & \multicolumn{2}{c|}{\alg} \\
    \cline{3-14}
          &       & IID   & Non-IID & IID   & Non-IID & IID   & Non-IID & IID   & Non-IID & IID   & Non-IID & IID   & Non-IID \\
    \hline
    \multirow{4}{*}{Texas100} & Passive & 0.643 & 0.623 & 0.646 & 0.620 & 0.650 & 0.623 & 0.560 & 0.550 & 0.633 & 0.606 & 0.616 & 0.596 \\
          & GA & 0.823 & 0.790 & 0.813 & 0.783 & 0.816 & 0.780 & 0.736 & 0.720 & 0.813 & 0.756 & 0.800 & 0.596 \\
          & AGREvader &  0.743 & 0.756 & 0.741 & 0.777 & 0.830 & 0.850 & 0.803 & 0.813 & 0.757 & 0.771 & 0.751 & 0.764  \\
          & \cellcolor{greyL}FedPoisonMIA &  \cellcolor{greyL}0.894 & \cellcolor{greyL}0.914 & \cellcolor{greyL}0.893 & \cellcolor{greyL}0.906 & \cellcolor{greyL}0.893 & \cellcolor{greyL}0.900 & \cellcolor{greyL}0.923 & \cellcolor{greyL}0.950 & \cellcolor{greyL}0.887 & \cellcolor{greyL}0.890 & \cellcolor{greyL}0.834 & \cellcolor{greyL}0.883  \\
    \hline
    \multirow{4}{*}{CIFAR-10} & Passive & 0.607 & 0.603 & 0.600 & 0.590 & 0.580 & 0.600 & 0.567 & 0.570 & 0.577 & 0.563 & 0.583 & 0.583 \\
          & GA & 0.737 & 0.603 & 0.733 & 0.610 & 0.723 & 0.630 & 0.693 & 0.590 & 0.707 & 0.597 & 0.710 & 0.603 \\
\hline
    \end{tabular}%
\end{table*}

\begin{table*}[htbp]
  \centering
    \footnotesize
  \addtolength{\tabcolsep}{-1.6pt}
  \renewcommand{\arraystretch}{1.2}
  \caption{Attack accuracy results in asynchronous setting.}
   \subfloat[$C=0.8$]
   {
    \renewcommand{\arraystretch}{1.2}
    \begin{tabular}{|c|c|cc|cc|cc|cc|cc|cc|}
    \hline
    \multirow{2}{*}{Dataset} & \multirow{2}{*}{Attack} & \multicolumn{2}{c|}{DP} & \multicolumn{2}{c|}{Top-$k$} & \multicolumn{2}{c|}{FedAvg} & \multicolumn{2}{c|}{Median} & \multicolumn{2}{c|}{Trimmed-mean} & \multicolumn{2}{c|}{\alg} \\
\cline{3-14}          &       & IID   & Non-IID & IID   & Non-IID & IID   & Non-IID & IID   & Non-IID & IID   & Non-IID & IID   & Non-IID \\
    \hline
\multirow{4}{*}{Texas100} & Passive & 0.653 & 0.636 & 0.660 & 0.640 & 0.663 & 0.633 & 0.610 & 0.613 & 0.650 & 0.620 & 0.646 & 0.640 \\
                              & GA &  0.790 & 0.750 & 0.790 & 0.740 & 0.790 & 0.763 & 0.733 & 0.750 & 0.786 & 0.743 & 0.630 & 0.640 \\
                              & AGREvader & 0.743 & 0.726 & 0.736 & 0.753 & 0.803 & 0.820 & 0.748 & 0.760 & 0.746 & 0.740 & 0.726 & 0.740 \\
                              & \cellcolor{greyL}FedPoisonMIA   &  \cellcolor{greyL}0.843 & \cellcolor{greyL}0.867 & \cellcolor{greyL}0.871 & \cellcolor{greyL}0.881 & \cellcolor{greyL}0.863 & \cellcolor{greyL}0.863 & \cellcolor{greyL}0.804 & \cellcolor{greyL}0.853 & \cellcolor{greyL}0.806 & \cellcolor{greyL}0.843 & \cellcolor{greyL}0.736 & \cellcolor{greyL}0.750 \\

    \hline
\multirow{4}{*}{CIFAR-10} & Passive & 0.597 & 0.580 & 0.590 & 0.580 & 0.633 & 0.620 & 0.607 & 0.597 & 0.613 & 0.607 & 0.633 & 0.593 \\
                              & GA &  0.630 & 0.563 & 0.603 & 0.577 & 0.653 & 0.580 & 0.630 & 0.557 & 0.633 & 0.580 & 0.583 & 0.577 \\
\hline
    \end{tabular}%
    }
    \vspace{0.08in}
    \\
     \subfloat[$C=1.0$]
   {
   
    \renewcommand{\arraystretch}{1.2}
    \begin{tabular}{|c|c|cc|cc|cc|cc|cc|cc|}
    \hline
    \multirow{2}{*}{Dataset} & \multirow{2}{*}{Attack} & \multicolumn{2}{c|}{DP} & \multicolumn{2}{c|}{Top-$k$} & \multicolumn{2}{c|}{FedAvg} & \multicolumn{2}{c|}{Median} & \multicolumn{2}{c|}{Trimmed-mean} & \multicolumn{2}{c|}{\alg} \\
\cline{3-14}          &       & IID   & Non-IID & IID   & Non-IID & IID   & Non-IID & IID   & Non-IID & IID   & Non-IID & IID   & Non-IID \\
    \hline
    \multirow{4}{*}{Texas100} & Passive & 0.670 & 0.653 & 0.673 & 0.656 & 0.670 & 0.650 & 0.593 & 0.600 & 0.650 & 0.633 & 0.600 & 0.660 \\
                              & GA & 0.806 & 0.786 & 0.806 & 0.783 & 0.803 & 0.776 & 0.746 & 0.683 & 0.773 & 0.770 & 0.613 & 0.683 \\
                              & AGREvader & 0.750 & 0.736 & 0.743 & 0.740 & 0.794 & 0.850 & 0.757 & 0.760 & 0.750 & 0.751 & 0.739 & 0.746  \\
                              & \cellcolor{greyL}FedPoisonMIA   & \cellcolor{greyL}0.877 & \cellcolor{greyL}0.900 &  \cellcolor{greyL}0.881 &  \cellcolor{greyL}0.897 &  \cellcolor{greyL}0.887 &  \cellcolor{greyL}0.904 &  \cellcolor{greyL}0.903 &  \cellcolor{greyL}0.930 &  \cellcolor{greyL}0.817 &  \cellcolor{greyL}0.877 &  \cellcolor{greyL}0.733 &  \cellcolor{greyL}0.743  \\
    \hline
    \multirow{4}{*}{CIFAR-10} & Passive & 0.620 & 0.577 & 0.593 & 0.570 & 0.613 & 0.593 & 0.587 & 0.587 & 0.597 & 0.587 & 0.620 & 0.587 \\
                              & GA & 0.637 & 0.597 & 0.653 & 0.590 & 0.657 & 0.570 & 0.660 & 0.590 & 0.670 & 0.587 & 0.587 & 0.550 \\
\hline
    \end{tabular}%
    }
  \label{asyn_attack_acc_result}%
\end{table*}%

\begin{table*}[htbp]
  \centering
    \footnotesize
  \addtolength{\tabcolsep}{-1.6pt}
  \renewcommand{\arraystretch}{1.2}
  \caption{Attack precision results in synchronous setting.}
   \subfloat[$C=0.8$]
   {
    \renewcommand{\arraystretch}{1.2} 
    \begin{tabular}{|c|c|cc|cc|cc|cc|cc|cc|}
    \hline
    \multirow{2}{*}{Dataset} & \multirow{2}{*}{Attack} & \multicolumn{2}{c|}{DP} & \multicolumn{2}{c|}{Top-$k$} & \multicolumn{2}{c|}{FedAvg} & \multicolumn{2}{c|}{Median} & \multicolumn{2}{c|}{Trimmed-mean} & \multicolumn{2}{c|}{\alg} \\
\cline{3-14}          &       & IID   & Non-IID & IID   & Non-IID & IID   & Non-IID & IID   & Non-IID & IID   & Non-IID & IID   & Non-IID \\
    \hline
    \multirow{4}{*}{Texas100} & Passive & 0.593 & 0.579 & 0.590 & 0.577 & 0.578 & 0.565 & 0.553 & 0.542 & 0.578 & 0.565 & 0.559 & 0.568 \\
                              & GA & 0.750 & 0.711 & 0.748 & 0.679 & 0.750 & 0.736 & 0.767 & 0.755 & 0.741 & 0.725 & 0.680 & 0.707 \\
                              & AGREvader & 0.685 & 0.684 & 0.685 & 0.683 & 0.719 & 0.799 & 0.677 & 0.668 & 0.674 & 0.683 & 0.659 & 0.690 \\
                              &  \cellcolor{greyL}FedPoisonMIA &  \cellcolor{greyL}0.821 &  \cellcolor{greyL}0.825 &  \cellcolor{greyL}0.821 &  \cellcolor{greyL}0.844 &  \cellcolor{greyL}0.830 &  \cellcolor{greyL}0.830 &  \cellcolor{greyL}0.853 &  \cellcolor{greyL}0.891 &  \cellcolor{greyL}0.807 &  \cellcolor{greyL}0.816 &  \cellcolor{greyL}0.719 &  \cellcolor{greyL}0.774 \\
    \hline
    \multirow{4}{*}{CIFAR-10} & Passive & 0.548 & 0.547 & 0.539 & 0.544 & 0.545 & 0.539 & 0.583 & 0.598 & 0.564 & 0.562 & 0.579 & 0.559 \\
                              & GA & 0.633 & 0.609 & 0.649 & 0.585 & 0.659 & 0.583 & 0.660 & 0.590 & 0.659 & 0.624 & 0.590 & 0.585 \\
\hline
    \end{tabular}%
    }
    \vspace{0.08in}
    \\
     \subfloat[$C=1.0$ ]
   {
    \renewcommand{\arraystretch}{1.2}
    \begin{tabular}{|c|c|cc|cc|cc|cc|cc|cc|}
    \hline
    \multirow{2}{*}{Dataset} & \multirow{2}{*}{Attack} & \multicolumn{2}{c|}{DP} & \multicolumn{2}{c|}{Top-$k$} & \multicolumn{2}{c|}{FedAvg} & \multicolumn{2}{c|}{Median} & \multicolumn{2}{c|}{Trimmed-mean} & \multicolumn{2}{c|}{\alg} \\
\cline{3-14}          &       & IID   & Non-IID & IID   & Non-IID & IID   & Non-IID & IID   & Non-IID & IID   & Non-IID & IID   & Non-IID \\
    \hline
    \multirow{4}{*}{Texas100} & Passive & 0.587 & 0.574 & 0.589 & 0.578 & 0.583 & 0.566 & 0.535 & 0.533 & 0.583 & 0.566 & 0.570 & 0.559 \\
                              & GA & 0.754 & 0.743 & 0.745 & 0.732 & 0.749 & 0.728 & 0.755 & 0.775 & 0.749 & 0.733 & 0.709 & 0.719 \\
                              & AGREvader & 0.664 & 0.682 & 0.659 & 0.693 & 0.751 & 0.810 & 0.722 & 0.683 & 0.674 & 0.686 & 0.668 & 0.680 \\
                              &  \cellcolor{greyL}FedPoisonMIA   &  \cellcolor{greyL}0.825 &  \cellcolor{greyL}0.844 &  \cellcolor{greyL}0.825 &  \cellcolor{greyL}0.853 &  \cellcolor{greyL}0.825 &  \cellcolor{greyL}0.834 &  \cellcolor{greyL}0.868 &  \cellcolor{greyL}0.919 &  \cellcolor{greyL}0.816 &  \cellcolor{greyL}0.821 &  \cellcolor{greyL}0.751 &  \cellcolor{greyL}0.812 \\
    \hline
    \multirow{4}{*}{CIFAR-10} & Passive & 0.558 & 0.538 & 0.560 & 0.568 & 0.547 & 0.600 & 0.573 & 0.590 & 0.543 & 0.576 & 0.547 & 0.549 \\
                              & GA & 0.662 & 0.585 & 0.662 & 0.586 & 0.654 & 0.595 & 0.683 & 0.597 & 0.646 & 0.574 & 0.638 & 0.589 \\
\hline
    \end{tabular}%
    }
  \label{precisoion_result}%
\end{table*}%

\begin{table*}[htbp]
  \centering
    \footnotesize
  \addtolength{\tabcolsep}{-1.6pt}
  \renewcommand{\arraystretch}{1.2}
  \caption{Attack precision results in asynchronous setting.}
   \subfloat[$C=0.8$]
   {
    \renewcommand{\arraystretch}{1.2} 
    \begin{tabular}{|c|c|cc|cc|cc|cc|cc|cc|}
    \hline
    \multirow{2}{*}{Dataset} & \multirow{2}{*}{Attack} & \multicolumn{2}{c|}{DP} & \multicolumn{2}{c|}{Top-$k$} & \multicolumn{2}{c|}{FedAvg} & \multicolumn{2}{c|}{Median} & \multicolumn{2}{c|}{Trimmed-mean} & \multicolumn{2}{c|}{\alg} \\
\cline{3-14}          &       & IID   & Non-IID & IID   & Non-IID & IID   & Non-IID & IID   & Non-IID & IID   & Non-IID & IID   & Non-IID \\
    \hline
    \multirow{4}{*}{Texas100} & Passive & 0.596 & 0.595 & 0.600 & 0.594 & 0.590 & 0.579 & 0.583 & 0.615 & 0.590 & 0.579 & 0.608 & 0.648 \\
                              & GA & 0.723 & 0.702 & 0.719 & 0.688 & 0.719 & 0.691 & 0.711 & 0.727 & 0.724 & 0.709 & 0.627 & 0.618 \\
                              & AGREvader & 0.680 & 0.660 & 0.664 & 0.684 & 0.738 & 0.740 & 0.653 & 0.673 & 0.675 & 0.661 & 0.675 & 0.671 \\
                              &  \cellcolor{greyL}FedPoisonMIA &  \cellcolor{greyL}0.764 &  \cellcolor{greyL}0.791 &  \cellcolor{greyL}0.795 &  \cellcolor{greyL}0.807 &  \cellcolor{greyL}0.786 &  \cellcolor{greyL}0.791 &  \cellcolor{greyL}0.719 &  \cellcolor{greyL}0.774 &  \cellcolor{greyL}0.723 &  \cellcolor{greyL}0.763 &  \cellcolor{greyL}0.656 &  \cellcolor{greyL}0.673 \\
    \hline
    \multirow{4}{*}{CIFAR-10} & Passive & 0.566 & 0.561 & 0.557 & 0.603 & 0.583 & 0.571 & 0.565 & 0.590 & 0.566 & 0.567 & 0.578 & 0.608 \\
                              & GA & 0.584 & 0.574 & 0.568 & 0.591 & 0.602 & 0.560 & 0.601 & 0.584 & 0.588 & 0.562 & 0.574 & 0.577 \\
\hline
    \end{tabular}%
    }
     \vspace{0.08in}
    \\
     \subfloat[$C=1.0$ ]
   {
    \renewcommand{\arraystretch}{1.2}
    \begin{tabular}{|c|c|cc|cc|cc|cc|cc|cc|}
    \hline
    \multirow{2}{*}{Dataset} & \multirow{2}{*}{Attack} & \multicolumn{2}{c|}{DP} & \multicolumn{2}{c|}{Top-$k$} & \multicolumn{2}{c|}{FedAvg} & \multicolumn{2}{c|}{Median} & \multicolumn{2}{c|}{Trimmed-mean} & \multicolumn{2}{c|}{\alg} \\
\cline{3-14}          &       & IID   & Non-IID & IID   & Non-IID & IID   & Non-IID & IID   & Non-IID & IID   & Non-IID & IID   & Non-IID \\
    \hline
    \multirow{4}{*}{Texas100} & Passive & 0.606 & 0.609 & 0.609 & 0.611 & 0.593 & 0.590 & 0.592 & 0.570 & 0.593 & 0.590 & 0.585 & 0.617 \\
                              & GA & 0.730 & 0.729 & 0.730 & 0.730 & 0.726 & 0.722 & 0.731 & 0.695 & 0.703 & 0.734 & 0.580 & 0.628 \\
                              & AGREvader & 0.676 & 0.670 & 0.670 & 0.673 & 0.709 & 0.772 & 0.674 & 0.679 & 0.671 & 0.668 & 0.651 & 0.674 \\
                              &  \cellcolor{greyL}FedPoisonMIA &  \cellcolor{greyL}0.803 &  \cellcolor{greyL}0.834 &  \cellcolor{greyL}0.807 &  \cellcolor{greyL}0.830 &  \cellcolor{greyL}0.816 &  \cellcolor{greyL}0.803 &  \cellcolor{greyL}0.831 &  \cellcolor{greyL}0.878 &  \cellcolor{greyL}0.733 &  \cellcolor{greyL}0.803 &  \cellcolor{greyL}0.675 &  \cellcolor{greyL}0.665 \\
    \hline
    \multirow{4}{*}{CIFAR-10} & Passive & 0.583 & 0.564 & 0.556 & 0.571 & 0.569 & 0.562 & 0.583 & 0.603 & 0.563 & 0.551 & 0.569 & 0.584 \\
                              & GA & 0.589 & 0.565 & 0.599 & 0.562 & 0.604 & 0.547 & 0.624 & 0.569 & 0.615 & 0.559 & 0.571 & 0.566 \\
\hline
    \end{tabular}%
    }
  \label{precisoion_result_2}%
\end{table*}%

\begin{table*}[htbp]
  \centering
    \footnotesize
  \addtolength{\tabcolsep}{-1.6pt}
  \renewcommand{\arraystretch}{1.2}
  \caption{Attack recall results in synchronous setting.}
   \subfloat[$C=0.8$]
   {
    \renewcommand{\arraystretch}{1.2}
    \begin{tabular}{|c|c|cc|cc|cc|cc|cc|cc|}
    \hline
    \multirow{2}{*}{Dataset} & \multirow{2}{*}{Attack} & \multicolumn{2}{c|}{DP} & \multicolumn{2}{c|}{Top-$k$} & \multicolumn{2}{c|}{FedAvg} & \multicolumn{2}{c|}{Median} & \multicolumn{2}{c|}{Trimmed-mean} & \multicolumn{2}{c|}{\alg} \\
\cline{3-14}          &       & IID   & Non-IID & IID   & Non-IID & IID   & Non-IID & IID   & Non-IID & IID   & Non-IID & IID   & Non-IID \\
    \hline
    \multirow{4}{*}{Texas100} & Passive & 0.960 & 0.927 & 0.960 & 0.920 & 0.960 & 0.933 & 0.867 & 0.853 & 0.960 & 0.933 & 0.940 & 0.860 \\
                              & GA & 0.980 & 0.953 & 0.987 & 0.947 & 0.980 & 0.893 & 0.920 & 0.800 & 0.953 & 0.913 & 0.893 & 0.933 \\
                              & AGREvader & 0.987 & 0.953 & 0.987 & 1.000 & 1.000 & 1.000 & 1.000 & 0.993 & 0.993 & 1.000 & 1.000 & 1.000 \\
                              &  \cellcolor{greyL}FedPoisonMIA &  \cellcolor{greyL}1.000 &  \cellcolor{greyL}1.000 &  \cellcolor{greyL}1.000 &  \cellcolor{greyL}1.000 &  \cellcolor{greyL}1.000 &  \cellcolor{greyL}1.000 &  \cellcolor{greyL}1.000 &  \cellcolor{greyL}0.980 &  \cellcolor{greyL}1.000 &  \cellcolor{greyL}1.000 &  \cellcolor{greyL}1.000 &  \cellcolor{greyL}1.000 \\
    \hline
    \multirow{4}{*}{CIFAR-10} & Passive & 0.993 & 0.927 & 0.960 & 0.873 & 0.960 & 0.960 & 0.420 & 0.447 & 0.527 & 0.540 & 1.000 & 0.853 \\
                              & GA & 0.933 & 0.780 & 0.927 & 0.780 & 0.940 & 0.773 & 0.907 & 0.633 & 0.940 & 0.787 & 1.000 & 0.847 \\
\hline
    \end{tabular}%
    }
    \vspace{0.08in}
    \\
     \subfloat[$C=1.0$]
   {
    \renewcommand{\arraystretch}{1.2}
    \begin{tabular}{|c|c|cc|cc|cc|cc|cc|cc|}
    \hline
    \multirow{2}{*}{Dataset} & \multirow{2}{*}{Attack} & \multicolumn{2}{c|}{DP} & \multicolumn{2}{c|}{Top-$k$} & \multicolumn{2}{c|}{FedAvg} & \multicolumn{2}{c|}{Median} & \multicolumn{2}{c|}{Trimmed-mean} & \multicolumn{2}{c|}{\alg} \\
\cline{3-14}          &       & IID   & Non-IID & IID   & Non-IID & IID   & Non-IID & IID   & Non-IID & IID   & Non-IID & IID   & Non-IID \\
    \hline
    \multirow{4}{*}{Texas100} & Passive & 0.967 & 0.953 & 0.967 & 0.893 & 0.933 & 0.913 & 0.913 & 0.813 & 0.933 & 0.913 & 0.947 & 0.913 \\
                              & GA & 0.960 & 0.887 & 0.953 & 0.893 & 0.953 & 0.893 & 0.700 & 0.620 & 0.953 & 0.807 & 0.927 & 0.763 \\
                              & AGREvader & 0.987 & 0.960 & 1.000 & 1.000 & 1.000 & 0.993 & 0.987 & 0.833 & 1.000 & 1.000 & 1.000 & 1.000 \\
                              &  \cellcolor{greyL}FedPoisonMIA &  \cellcolor{greyL}1.000 &  \cellcolor{greyL}1.000 &  \cellcolor{greyL}1.000 &  \cellcolor{greyL}1.000 &  \cellcolor{greyL}1.000 &  \cellcolor{greyL}1.000 &  \cellcolor{greyL}1.000 &  \cellcolor{greyL}0.987 &  \cellcolor{greyL}1.000 &  \cellcolor{greyL}1.000 &  \cellcolor{greyL}1.000 &  \cellcolor{greyL}1.000 \\
    \hline
    \multirow{4}{*}{CIFAR-10} & Passive & 0.960 & 0.987 & 0.527 & 0.527 & 0.927 & 0.600 & 0.527 & 0.460 & 0.960 & 0.480 & 0.973 & 0.940 \\
                              & GA & 0.967 & 0.713 & 0.953 & 0.747 & 0.947 & 0.813 & 0.720 & 0.553 & 0.913 & 0.747 & 0.973 & 0.687 \\
\hline
    \end{tabular}%
    }
  \label{recall_result}%
\end{table*}%

\begin{table*}[htbp]
  \centering
    \footnotesize
  \addtolength{\tabcolsep}{-1.6pt}
  \renewcommand{\arraystretch}{1.2}
  \caption{Attack recall results in asynchronous setting.}
   \subfloat[$C=0.8$]
   {
    \renewcommand{\arraystretch}{1.2} 
    \begin{tabular}{|c|c|cc|cc|cc|cc|cc|cc|}
    \hline
    \multirow{2}{*}{Dataset} & \multirow{2}{*}{Attack} & \multicolumn{2}{c|}{DP} & \multicolumn{2}{c|}{Top-$k$} & \multicolumn{2}{c|}{FedAvg} & \multicolumn{2}{c|}{Median} & \multicolumn{2}{c|}{Trimmed-mean} & \multicolumn{2}{c|}{\alg} \\
\cline{3-14}          &       & IID   & Non-IID & IID   & Non-IID & IID   & Non-IID & IID   & Non-IID & IID   & Non-IID & IID   & Non-IID \\
    \hline
    \multirow{4}{*}{Texas100} & Passive & 0.953 & 0.853 & 0.960 & 0.887 & 0.980 & 0.880 & 0.773 & 0.607 & 0.980 & 0.880 & 0.827 & 0.613 \\
                              & GA & 0.940 & 0.893 & 0.953 & 0.880 & 0.953 & 0.880 & 0.920 & 0.800 & 0.927 & 0.827 & 0.640 & 0.733 \\
                              & AGREvader & 0.920 & 0.933 & 0.960 & 0.940 & 0.940 & 0.967 & 0.921 & 0.987 & 0.873 & 0.987 & 0.873 & 0.940 \\
                              &  \cellcolor{greyL}FedPoisonMIA &  \cellcolor{greyL}0.993 &  \cellcolor{greyL}1.000 &  \cellcolor{greyL}1.000 &  \cellcolor{greyL}1.000 &  \cellcolor{greyL}1.000 &  \cellcolor{greyL}0.987 &  \cellcolor{greyL}1.000 &  \cellcolor{greyL}1.000 &  \cellcolor{greyL}0.993 &  \cellcolor{greyL}1.000 &  \cellcolor{greyL}0.993 &  \cellcolor{greyL}0.973 \\
    \hline
    \multirow{4}{*}{CIFAR-10} & Passive & 0.827 & 0.740 & 0.880 & 0.467 & 0.940 & 0.960 & 0.933 & 0.633 & 0.973 & 0.900 & 0.987 & 0.527 \\
                              & GA & 0.900 & 0.467 & 0.867 & 0.500 & 0.907 & 0.747 & 0.773 & 0.393 & 0.893 & 0.727 & 0.647 & 0.573 \\
\hline
    \end{tabular}%
    }
     \vspace{0.08in}
    \\
     \subfloat[$C=1.0$]
   {
    \renewcommand{\arraystretch}{1.2}
    \begin{tabular}{|c|c|cc|cc|cc|cc|cc|cc|}
    \hline
    \multirow{2}{*}{Dataset} & \multirow{2}{*}{Attack} & \multicolumn{2}{c|}{DP} & \multicolumn{2}{c|}{Top-$k$} & \multicolumn{2}{c|}{FedAvg} & \multicolumn{2}{c|}{Median} & \multicolumn{2}{c|}{Trimmed-mean} & \multicolumn{2}{c|}{\alg} \\
\cline{3-14}          &       & IID   & Non-IID & IID   & Non-IID & IID   & Non-IID & IID   & Non-IID & IID   & Non-IID & IID   & Non-IID \\
    \hline
    \multirow{4}{*}{Texas100} & Passive & 0.973 & 0.860 & 0.967 & 0.860 & 0.953 & 0.873 & 0.600 & 0.813 & 0.953 & 0.873 & 0.687 & 0.880 \\
                              & GA & 0.973 & 0.913 & 0.973 & 0.900 & 0.973 & 0.900 & 0.780 & 0.653 & 0.947 & 0.847 & 0.820 & 0.900 \\
                              & AGREvader & 0.960 & 0.933 & 0.960 & 0.933 & 1.000 & 0.993 & 1.000 & 0.987 & 0.980 & 1.000 & 0.957 & 0.953 \\
                              &  \cellcolor{greyL}FedPoisonMIA &  \cellcolor{greyL}1.000 &  \cellcolor{greyL}1.000 &  \cellcolor{greyL}1.000 &  \cellcolor{greyL}1.000 &  \cellcolor{greyL}1.000 &  \cellcolor{greyL}1.000 &  \cellcolor{greyL}1.000 &  \cellcolor{greyL}1.000 &  \cellcolor{greyL}1.000 &  \cellcolor{greyL}1.000 &  \cellcolor{greyL}0.900 &  \cellcolor{greyL}0.980 \\
    \hline
    \multirow{4}{*}{CIFAR-10} & Passive & 0.840 & 0.673 & 0.933 & 0.560 & 0.940 & 0.847 & 0.607 & 0.507 & 0.860 & 0.940 & 0.993 & 0.600 \\
                              & GA & 0.907 & 0.840 & 0.927 & 0.820 & 0.907 & 0.813 & 0.807 & 0.747 & 0.907 & 0.820 & 0.700 & 0.427 \\
\hline
    \end{tabular}%
    }
  \label{recall_result_2}%
\end{table*}%

\begin{table*}[htbp]
  \centering
    \footnotesize
  \addtolength{\tabcolsep}{-1.6pt}
  \renewcommand{\arraystretch}{1.2}
  \caption{Test accuracy results in synchronous setting.}
   \subfloat[$C=0.8$]
   {
    \renewcommand{\arraystretch}{1.2}
    \begin{tabular}{|c|c|cc|cc|cc|cc|cc|cc|}
    \hline
    \multirow{2}{*}{Dataset} & \multirow{2}{*}{Attack} & \multicolumn{2}{c|}{DP} & \multicolumn{2}{c|}{Top-$k$} & \multicolumn{2}{c|}{FedAvg} & \multicolumn{2}{c|}{Median} & \multicolumn{2}{c|}{Trimmed-mean} & \multicolumn{2}{c|}{\alg} \\
\cline{3-14}          &       & IID   & Non-IID & IID   & Non-IID & IID   & Non-IID & IID   & Non-IID & IID   & Non-IID & IID   & Non-IID \\
    \hline
    \multirow{4}{*}{Texas100} & Passive & 0.577 & 0.570 & 0.577 & 0.569 & 0.579 & 0.571 & 0.583 & 0.570 & 0.579 & 0.571 & 0.583 & 0.566 \\
                              & GA & 0.572 & 0.565 & 0.573 & 0.573 & 0.573 & 0.566 & 0.572 & 0.563 & 0.571 & 0.566 & 0.578 & 0.567 \\
                              & AGREvader & 0.576 & 0.567 & 0.575 & 0.574 & 0.572 & 0.568 & 0.574 & 0.558 & 0.578 & 0.570 & 0.580 & 0.568 \\
                              &  \cellcolor{greyL}FedPoisonMIA &  \cellcolor{greyL}0.577 &  \cellcolor{greyL}0.565 &  \cellcolor{greyL}0.573 &  \cellcolor{greyL}0.562 &  \cellcolor{greyL}0.575 &  \cellcolor{greyL}0.563 &  \cellcolor{greyL}0.566 &  \cellcolor{greyL}0.544 &  \cellcolor{greyL}0.571 &  \cellcolor{greyL}0.555 &  \cellcolor{greyL}0.576 &  \cellcolor{greyL}0.569 \\
    \hline
    \multirow{4}{*}{CIFAR-10} & Passive & 0.790 & 0.741 & 0.766 & 0.741 & 0.765 & 0.741 & 0.727 & 0.677 & 0.768 & 0.740 & 0.785 & 0.754 \\
                              & GA & 0.763 & 0.730 & 0.765 & 0.728 & 0.765 & 0.728 & 0.738 & 0.716 & 0.760 & 0.729 & 0.766 & 0.719 \\
\hline
    \end{tabular}%
    }
    \vspace{0.08in}
    \\
     \subfloat[$C=1.0$]
   {
    \renewcommand{\arraystretch}{1.2}
    \begin{tabular}{|c|c|cc|cc|cc|cc|cc|cc|}
    \hline
    \multirow{2}{*}{Dataset} & \multirow{2}{*}{Attack} & \multicolumn{2}{c|}{DP} & \multicolumn{2}{c|}{Top-$k$} & \multicolumn{2}{c|}{FedAvg} & \multicolumn{2}{c|}{Median} & \multicolumn{2}{c|}{Trimmed-mean} & \multicolumn{2}{c|}{\alg} \\
\cline{3-14}          &       & IID   & Non-IID & IID   & Non-IID & IID   & Non-IID & IID   & Non-IID & IID   & Non-IID & IID   & Non-IID \\
    \hline
    \multirow{4}{*}{Texas100} & Passive & 0.575 & 0.571 & 0.575 & 0.569 & 0.577 & 0.570 & 0.583 & 0.572 & 0.577 & 0.570 & 0.580 & 0.575 \\
                              & GA & 0.572 & 0.567 & 0.571 & 0.567 & 0.572 & 0.565 & 0.573 & 0.558 & 0.572 & 0.565 & 0.573 & 0.563 \\
                              & AGREvader & 0.574 & 0.569 & 0.578 & 0.568 & 0.570 & 0.569 & 0.571 & 0.568 & 0.577 & 0.567 & 0.579 & 0.569 \\
                              &  \cellcolor{greyL}FedPoisonMIA &  \cellcolor{greyL}0.571 &  \cellcolor{greyL}0.562 &  \cellcolor{greyL}0.575 &  \cellcolor{greyL}0.560 &  \cellcolor{greyL}0.570 &  \cellcolor{greyL}0.561 &  \cellcolor{greyL}0.563 &  \cellcolor{greyL}0.541 &  \cellcolor{greyL}0.569 &  \cellcolor{greyL}0.554 &  \cellcolor{greyL}0.575 &  \cellcolor{greyL}0.563 \\
    \hline
    \multirow{4}{*}{CIFAR-10} & Passive & 0.766 & 0.736 & 0.772 & 0.732 & 0.772 & 0.732 & 0.756 & 0.685 & 0.771 & 0.721 & 0.774 & 0.750 \\
                              & GA & 0.755 & 0.730 & 0.758 & 0.735 & 0.767 & 0.733 & 0.739 & 0.714 & 0.756 & 0.726 & 0.756 & 0.719 \\
\hline
    \end{tabular}%
    }
  \label{test_acc_result}%
\end{table*}%

\begin{table*}[htbp]
  \centering
    \footnotesize
  \addtolength{\tabcolsep}{-1.6pt}
  \renewcommand{\arraystretch}{1.2}
  \caption{Test accuracy results in asynchronous setting.}
   \subfloat[$C=0.8$]
   {
    \renewcommand{\arraystretch}{1.2}
    \begin{tabular}{|c|c|cc|cc|cc|cc|cc|cc|}
    \hline
    \multirow{2}{*}{Dataset} & \multirow{2}{*}{Attack} & \multicolumn{2}{c|}{DP} & \multicolumn{2}{c|}{Top-$k$} & \multicolumn{2}{c|}{FedAvg} & \multicolumn{2}{c|}{Median} & \multicolumn{2}{c|}{Trimmed-mean} & \multicolumn{2}{c|}{\alg} \\
\cline{3-14}          &       & IID   & Non-IID & IID   & Non-IID & IID   & Non-IID & IID   & Non-IID & IID   & Non-IID & IID   & Non-IID \\
    \hline
    \multirow{4}{*}{Texas100} & Passive & 0.570 & 0.572 & 0.567 & 0.570 & 0.575 & 0.572 & 0.588 & 0.566 & 0.575 & 0.572 & 0.552 & 0.516 \\
                              & GA & 0.569 & 0.566 & 0.569 & 0.567 & 0.570 & 0.566 & 0.583 & 0.563 & 0.571 & 0.565 & 0.511 & 0.465 \\
                              & AGREvader & 0.565 & 0.575 & 0.561 & 0.571 & 0.584 & 0.578 & 0.581 & 0.571 & 0.579 & 0.587 & 0.579 & 0.476 \\
                              &  \cellcolor{greyL}FedPoisonMIA &  \cellcolor{greyL}0.593 &  \cellcolor{greyL}0.581 &  \cellcolor{greyL}0.557 &  \cellcolor{greyL}0.581 &  \cellcolor{greyL}0.586 &  \cellcolor{greyL}0.579 &  \cellcolor{greyL}0.583 &  \cellcolor{greyL}0.556 &  \cellcolor{greyL}0.592 &  \cellcolor{greyL}0.581 &  \cellcolor{greyL}0.583 &  \cellcolor{greyL}0.571 \\
    \hline
    \multirow{4}{*}{CIFAR-10} & Passive & 0.743 & 0.706 & 0.776 & 0.717 & 0.783 & 0.763 & 0.713 & 0.679 & 0.779 & 0.752 & 0.748 & 0.645 \\
                              & GA & 0.763 & 0.698 & 0.763 & 0.712 & 0.749 & 0.706 & 0.733 & 0.650 & 0.759 & 0.702 & 0.735 & 0.608 \\
\hline
    \end{tabular}%
    }
    \vspace{0.08in}
    \\
     \subfloat[$C=1.0$]
   {
    \renewcommand{\arraystretch}{1.2}
    \begin{tabular}{|c|c|cc|cc|cc|cc|cc|cc|}
    \hline
    \multirow{2}{*}{Dataset} & \multirow{2}{*}{Attack} & \multicolumn{2}{c|}{DP} & \multicolumn{2}{c|}{Top-$k$} & \multicolumn{2}{c|}{FedAvg} & \multicolumn{2}{c|}{Median} & \multicolumn{2}{c|}{Trimmed-mean} & \multicolumn{2}{c|}{\alg} \\
\cline{3-14}          &       & IID   & Non-IID & IID   & Non-IID & IID   & Non-IID & IID   & Non-IID & IID   & Non-IID & IID   & Non-IID \\
    \hline
    \multirow{4}{*}{Texas100} & Passive & 0.573 & 0.578 & 0.570 & 0.578 & 0.575 & 0.584 & 0.585 & 0.586 & 0.575 & 0.584 & 0.577 & 0.561 \\
                              & GA & 0.570 & 0.572 & 0.569 & 0.574 & 0.568 & 0.576 & 0.582 & 0.579 & 0.570 & 0.581 & 0.567 & 0.551 \\
                              & AGREvader & 0.559 & 0.568 & 0.558 & 0.570 & 0.573 & 0.576 & 0.592 & 0.575 & 0.599 & 0.588 & 0.591 & 0.581 \\
                              &  \cellcolor{greyL}FedPoisonMIA &  \cellcolor{greyL}0.586 &  \cellcolor{greyL}0.580 &  \cellcolor{greyL}0.584 &  \cellcolor{greyL}0.582 &  \cellcolor{greyL}0.590 &  \cellcolor{greyL}0.577 &  \cellcolor{greyL}0.567 &  \cellcolor{greyL}0.554 &  \cellcolor{greyL}0.592 &  \cellcolor{greyL}0.576 &  \cellcolor{greyL}0.592 &  \cellcolor{greyL}0.591 \\
    \hline
    \multirow{4}{*}{CIFAR-10} & Passive & 0.747 & 0.714 & 0.765 & 0.754 & 0.780 & 0.741 & 0.726 & 0.667 & 0.792 & 0.761 & 0.745 & 0.669 \\
                              & GA & 0.731 & 0.719 & 0.765 & 0.746 & 0.749 & 0.726 & 0.734 & 0.735 & 0.750 & 0.724 & 0.746 & 0.590 \\
\hline
    \end{tabular}%
    }
  \label{test_acc_result_2}%
\end{table*}%

\begin{table*}[t]
\centering
\footnotesize
\caption{Attack running time (seconds).}
\label{attack_cost}
\begin{tabular}{|c|c|c|}
 \hline
    Dataset   & AGREvader & \cellcolor{greyL}FedPoisonMIA \\
 \hline
    Texas100  & 1104.832  & \cellcolor{greyL}1669.040 \\
    CIFAR-10  & 3673.408  & \cellcolor{greyL}4404.120 \\
    STL10     & 3508.848  & \cellcolor{greyL}4316.448 \\
    FER2013   & 3173.304  & \cellcolor{greyL}3810.128 \\
 \hline
\end{tabular}
\end{table*}

\begin{table*}[t]
\footnotesize
\centering
\caption{Computational cost (in seconds) of different defenses.}
\label{defense_cost}
    \begin{tabular}
    {|c|c|c|c|c|c|c|}
     \hline
    \multicolumn{1}{|c|}{\multirow{2}{*}{Dataset}} & \multicolumn{2}{c|}{Median} & \multicolumn{2}{c|}{Trimmed-mean} & \multicolumn{2}{c|}{\alg} \\
\cline{2-7}          & \multicolumn{1}{c|}{10 clients} & \multicolumn{1}{c|}{50 clients} & \multicolumn{1}{c|}{10 clients} & \multicolumn{1}{c|}{50 clients} & \multicolumn{1}{c|}{10 clients} & \multicolumn{1}{c|}{50 clients} \\
        \hline
        Texas100  & 5.262 & 18.344 & 2.942 & 20.944 & 0.975 & 4.688 \\
        CIFAR-10  & 0.320 & 0.618  & 0.250 & 1.729  & 0.295 & 0.371 \\
        STL10     & 0.312 & 0.700  & 0.250 & 1.658  & 0.314 & 0.405 \\
        FER2013   & 0.311 & 0.759  & 0.242 & 1.646  & 0.282 & 0.372 \\
    \hline
    \end{tabular}
\end{table*}

\begin{algorithm}[t]
    \caption{\alg algorithm.}
    \label{ATM}
    \renewcommand{\algorithmicrequire}{\textbf{Input:}}
    \renewcommand{\algorithmicensure}{\textbf{Output:}}
    \begin{algorithmic}[1]
        \Require Gradients from $n$ clients: $\mathcal{G} = \{\bm{g}_1, \bm{g}_2, \dots, \bm{g}_n\}$, trim parameter $b$.
        \Ensure Aggregated gradient $\bar{\bm{g}}$.
        
        \State Initialize an $n \times n$ zero matrix $\bm{A}$ to record the angles between gradients.  
        \For{each gradient pair $(i, j)$ where $1 \leq i < j \leq n$}

        \State $\theta_{i, j} = arccos(\frac{\bm{g}_i \cdot \bm{g}_j}{\|\bm{g}_i\| \|\bm{g}_j\|})$
            
        \State $\bm{A}[i,j] \leftarrow \theta_{i, j}$
         \EndFor

        \For{each row in $\bm{A}$}
            \State $\bar{\theta_i} = \frac{1}{n} \sum_{\substack{j=1}}^{n} \bm{A}[i,j] $
        \EndFor

        \State Discard the $2b$ gradients with the largest absolute values in $\bar{\theta}$; denote the remaining set as $\hat{\mathcal{G}}$.
        
        \State Calculate $\bar{g}$ by averaging the selected gradients as
         $\bar{\bm{g}} = \frac{1}{|\hat{\mathcal{G}}|} \sum_{ \bm{g} \in \hat{\mathcal{G}}} \bm{g}$

        \State Send aggregated gradient $\bar{\bm{g}}$ to clients.
    \end{algorithmic}
\end{algorithm}

\begin{algorithm}[t]
    \caption{Adaptive attack against \alg.}
    \label{alg:adaptive_attack}
    \renewcommand{\algorithmicrequire}{\textbf{Input:}}
    \renewcommand{\algorithmicensure}{\textbf{Output:}}
    \begin{algorithmic}[1]
        \Require Total number of clients $n$, a set of benign gradients $\mathcal{G}_\mathcal{B}$, attack gradient $\bm{g}_{\text{attack}}$, trim parameter $b$.
        \Ensure Adjusted malicious gradient $\bm{g}_{\text{adaptive}}$.
        
        \Repeat
            \State Initialize an $n \times n$ zero matrix $\bm{A}$ to record the angles between gradients.  
            \For{gradient pair $(i, j)$ where $1 \leq i < j \leq n$}
    
            \State $\theta_{i, j} = arccos(\frac{\bm{g}_i \cdot \bm{g}_j}{\|\bm{g}_i\| \|\bm{g}_j\|})$
                
            \State $\bm{A}[i,j] \leftarrow \theta_{i, j}$, $\bm{A}[j, i] \leftarrow \theta_{i, j}$
             \EndFor
    
            \For{each row in $\bm{A}$}
                \State $\bar{\theta_i} = \frac{1}{n} \sum_{\substack{j=1}}^{n} \bm{A}[i,j] $
            \EndFor

            \State Let $\theta_{\text{attack}}$ represent the final value of $\bar{\theta}$, corresponding to the average angular deviation between the attack gradient and the benign gradients.
            \State Arrange $\bar{\theta}$ in ascending order and define the trimming threshold $\theta_{\tau}$ as the angle ranked $2b$-th from the largest in the sorted list.
            \If{ \( \theta_{\text{attack}} < \theta_{\tau} \) }
                \State \textbf{break} \Comment{Attack is considered successful}
            \Else
                \State Select the benign gradient $\bm{g}_k$ with the greatest angular deviation from $\bm{g}_{\text{attack}}$.
                \State Update \( \bm{g}_{\text{attack}} \leftarrow \frac{1}{2}(\bm{g}_{\text{attack}} + \bm{g}_k) \)
            \EndIf
            
        \Until{convergence} 
        
        \State \Return $\bm{g}_{\text{adaptive}} \gets \bm{g}_{\text{attack}}$
    \end{algorithmic}
\end{algorithm}

\end{document}